\documentclass[aps,pra,superscriptaddress,amsmath,amssymb,showpacs,floatfix,notitlepage]{revtex4-1}

\usepackage{amsmath}
\usepackage{amsthm}
\usepackage{latexsym}
\usepackage{amsfonts}
\usepackage{amssymb}
\usepackage{bbm,dsfont}
\usepackage{color}
\usepackage{graphicx}
\usepackage{enumerate}
\usepackage{subfigure}
\usepackage[unicode=true, breaklinks=false, pdfborder={0 0 1}, backref=false, colorlinks=true, linkcolor=blue, citecolor=blue]{hyperref}

\usepackage[normalem]{ulem}
\usepackage{mathrsfs}
\usepackage{mathtools}

\usepackage{comment}

\usepackage[T1]{fontenc}

\usepackage{tikz}
\usetikzlibrary{calc, arrows.meta, intersections, decorations.pathreplacing}

\usepackage{tikz-3dplot} 

\usepackage{bm}

\definecolor{darkgreen}{rgb}{0,0.5,0}
\definecolor{lightgreen}{rgb}{0,0.75,0}

\newtheorem{proposition}{Proposition}

\theoremstyle{definition}

\definecolor{darkgreen}{rgb}{0,0.4,0.3}
\definecolor{gray}{rgb}{0.5,0.5,0.5}

\newcommand{\Rb}{\mathbb{R}} 
\newcommand{\Cb}{\mathbb C} 
\newcommand{\ca}[1]{\mathcal{#1}} 
\newcommand{\mo}[1]{\left| #1 \right|}

\newcommand{\rme}{{\rm e}}

\newcommand{\hh}{\mathcal{H}} 
\newcommand{\ket}[1]{|#1\rangle} 
\newcommand{\kb}[2]{|#1\rangle\langle#2|} 
\newcommand{\no}[1]{\left\|#1\right\|} 
\newcommand{\tr}[1]{{\rm tr}\left[#1\right]} 
\newcommand{\id}{\mathbbm{1}} 

\newcommand{\va}{{\vec{a}}} 
\newcommand{\vb}{{\vec{b}}} 
\newcommand{\vc}{{\vec{c}}} 
\newcommand{\vd}{{\vec{d}}} 
\newcommand{\ve}{{\vec{e}}} 
\newcommand{\vu}{{\vec{u}}} 
\newcommand{\vv}{{\vec{v}}} 
\newcommand{\vf}{{\vec{f}}}
\newcommand{\vsigma}{\vec{\sigma}}

\newcommand{\A}{\mathsf{A}}
\newcommand{\B}{\mathsf{B}}
\newcommand{\Cc}{\mathsf{C}}
\newcommand{\D}{\mathsf{D}}
\newcommand{\E}{\mathsf{E}}
\newcommand{\F}{\mathsf{F}}
\newcommand{\M}{\mathsf{M}}
\newcommand{\N}{\mathsf{N}}
\newcommand{\J}{\mathsf{J}}
\newcommand{\Meas}{\mathcal{M}} 
\newcommand{\JMeas}{\mathcal{CM}} 
\newcommand{\IR}{\rm IR} 

\newcommand{\en}{\mathcal{E}} 
\newcommand{\PP}{\mathbb{P}} 

\newcommand{\Sec}{\mathcal{S}} 

\tdplotsetmaincoords{60}{110}

\pgfmathsetmacro{\rvec}{.8}
\pgfmathsetmacro{\thetavec}{30}
\pgfmathsetmacro{\phivec}{60}

\pgfmathsetmacro{\rvecQ}{.8}
\pgfmathsetmacro{\thetavecQ}{90}
\pgfmathsetmacro{\phivecQ}{-45}

\pgfmathsetmacro{\rvecT}{1}
\pgfmathsetmacro{\thetavecT}{90}
\pgfmathsetmacro{\phivecT}{45}

\pgfmathsetmacro{\rvecC}{1}
\pgfmathsetmacro{\thetavecC}{90}
\pgfmathsetmacro{\phivecC}{60}

\pgfmathsetmacro{\rvecD}{.8}
\pgfmathsetmacro{\thetavecD}{90}
\pgfmathsetmacro{\phivecD}{20}

\pgfmathsetmacro{\rvecF}{.8}
\pgfmathsetmacro{\thetavecF}{90}
\pgfmathsetmacro{\phivecF}{-20}

\pgfmathsetmacro{\rvecDD}{.65}
\pgfmathsetmacro{\thetavecDD}{100}
\pgfmathsetmacro{\phivecDD}{-18}

\pgfmathsetmacro{\rvecFF}{1}
\pgfmathsetmacro{\thetavecFF}{80}
\pgfmathsetmacro{\phivecFF}{+18}

\begin{document}

\title{Experimentally determining the incompatibility of two qubit measurements}

\author{Andrea Smirne}
\email{andrea.smirne@unimi.it}
\affiliation{Dipartimento di Fisica ``Aldo Pontremoli'', Universit{\`a} degli Studi di Milano, Via Celoria 16, 20133 Milano, Italy}
\affiliation{Istituto Nazionale di Fisica Nucleare, Sezione di Milano, Via Celoria 16, 20133 Milano, Italy}

\author{Simone Cialdi}
\affiliation{Dipartimento di Fisica ``Aldo Pontremoli'', Universit{\`a} degli Studi di Milano, Via Celoria 16, 20133 Milano, Italy}
\affiliation{Istituto Nazionale di Fisica Nucleare, Sezione di Milano, Via Celoria 16, 20133 Milano, Italy}

\author{Daniele Cipriani}
\affiliation{Dipartimento di Fisica ``Aldo Pontremoli'', Universit{\`a} degli Studi di Milano, Via Celoria 16, 20133 Milano, Italy}

\author{Claudio Carmeli}
\affiliation{Dipartimento di Ingegneria Meccanica, Energetica, Gestionale e dei Trasporti, Universit\`a di Genova, Via Magliotto 2, 17100 Savona, Italy}

\author{Alessandro Toigo}
\affiliation{Dipartimento di Matematica, Politecnico di Milano, Piazza Leonardo da Vinci 32, 20133 Milano, Italy}
\affiliation{Istituto Nazionale di Fisica Nucleare, Sezione di Milano, Via Celoria 16, 20133 Milano, Italy}

\author{Bassano Vacchini}
\email{bassano.vacchini@mi.infn.it}
\affiliation{Dipartimento di Fisica ``Aldo Pontremoli'', Universit{\`a} degli Studi di Milano, Via Celoria 16, 20133 Milano, Italy}
\affiliation{Istituto Nazionale di Fisica Nucleare, Sezione di Milano, Via Celoria 16, 20133 Milano, Italy}

\begin{abstract}
We describe and 
realize an experimental procedure for assessing the incompatibility of two qubit measurements. The experiment consists in a state discrimination task where either measurement is used according to some partial intermediate information. The success statistics of the task provides an upper bound for the amount of incompatibility of the two measurements, as it is quantified by means of their incompatibility robustness. For a broad class of unbiased and possibly noisy qubit measurements, one can make this upper bound 
coincide with the true value of the robustness by suitably tuning the preparation of the experiment. We demonstrate this fact in an optical setup, where the qubit states are encoded into the photons' polarization degrees of freedom, and incompatibility is directly
accessed by virtue of a refined control on the amplitude, phase and purity of the final projection stage of the measurements. Our work thus establishes the practical feasibility
of a recently proposed method for the detection of quantum incompatibility.
\end{abstract}

\maketitle

\section{Introduction}

Quantum incompatibility is one of the most striking and fascinating features of the quantum world. It roots in the very essence of quantum physics as a noncommutative probabilistic theory: in such a theory, there necessarily exist measurements that are mutually exclusive, in the sense that their statistics cannot be post-processed from any single and more informative joint measurement \cite{HeMiZi16,Guehne2021}.

Despite its statistical nature, however, quantum incompatibility is elusive in practical experiments. Indeed, being a no-go statement about the theory, its empirical verification can only be of indirect type. Usually, quantum incompatibility is detected by means of the violation of Bell-type or steering inequalities \cite{WoPeFe09,QuVeBr14,UoMoGu14}. These inequalities are violated only in entangled states, and thus generally may require a demanding preparation to be performed by the experimenter.

The situation considerably changed after some recent works pointed out that quantum incompatibility can be assessed also in a {\em state discrimination task with intermediate information} \cite{CaHeTo19,SkSuCa19,UoKrShYuGu19}. Such a task does not involve multipartite quantum systems, and for this reason it is easier to be performed in the laboratory. The task is actually a variation of the usual state discrimination, from which it differs only in the following extra step: before guessing the unknown state, the experimenter receives some classical information, and, based on it, he performs a measurement chosen from some predetermined possible alternatives \cite{BaWeWi08,GoWe10,CaHeTo18,CaHeTo21}. It then turns out that the probability of guessing the correct state is connected to the incompatibility of the measurements used in the experiment. More precisely, the success probability yields an upper bound for the {\em incompatibility robustness} of these measurements, a quantifier of incompatibility that is widely used in the literature~\cite{UoBuGuPe15,Haapasalo15,DeFaKa19}. The bound is device-dependent, meaning that the experimenter needs to have full control over the states to be discriminated. Nevertheless, if the states are suitably tuned, the bound actually coincides with the true value of the robustness for a quite broad class of measurements.

In this paper, we devise and experimentally demonstrate a concrete and versatile strategy to determine quantum incompatibility via a state discrimination task with intermediate information. For the experimental implementation, we employ a two-level system realized in a quantum optical setup. Qubit states of the system are encoded into the photon polarization degrees of freedom, while the measurements are realized by means of a spatial light modulator, capable of introducing both  any wanted phase and any programmed noise \cite{Smirne2011,Smirne2013,Cialdi2017a,Cialdi2019}, and also ensuring a high degree of control on the final polarization projections. For a large family of measurements, we show the agreement between the theoretical evaluation of the incompatibility robustness and its assessment via the experimentally-determined success statistics of our state discrimination task. Moreover, we explore a class of measurements whose incompatibility robustness still lacks an analytic derivation, but nonetheless can be upper bounded via the experimental procedure proposed here.


The realization of experiments aimed at testing and assessing quantum incompatibility is a very novel field of research, to the point that, up to our knowledge, only two recent articles addressed this topic before the present one. In \cite{AnMuChMiTaBo20}, the degree of incompatibility of two qubit measurements was evaluated by determining the success probability of a communication protocol based on random access code. In the proposed experiment, incompatibility was detected and quantified by means of the violation of a classical bound in the quantum version of the protocol \cite{AmLeMaOz08}. More directly related to our approach, the authors of \cite{Wu2021} employed a state discrimination task with intermediate information in order to determine the incompatibility robustness of two mutually unbiased bases in a three dimensional system. They considered noisy versions of a fixed pair of such bases, and assessed incompatibility under different levels of the noise. In the present work, we drop the mutual unbiasedness hypothesis, thus extending the results of \cite{Wu2021} to a considerably more general class of measurements.



The paper is organized as follows. Section \ref{sec:inco_quantifier} recalls the basic facts about quantum incompatibility and, in particular, the notion of incompatibility robustness and its relation with state discrimination tasks with intermediate information. We then focus on the specific state discrimination task of our experiment and on the noisy unbiased qubit measurements that constitute the target class of measurements.
Section \ref{sec:the} contains the main results of the paper. After illustrating the experimental apparatus, we perform our state discrimination task and, by means of its results, we empirically determine the incompatibility robustness of all pairs of equally noisy unbiased qubit measurements. Moreover, by considering qubit measurements with different amounts of noise, we demonstrate that the same procedure provides a nontrivial upper bound for the incompatibility robustness even in circumstances where its analytic evaluation is still an open problem. Finally, Section \ref{sec:conc} discusses the conclusions and future outlooks of our work.

\section{Detecting quantum incompatibility}\label{sec:inco_quantifier}


\subsection{Incompatibility robustness}\label{sec:inco_robu}

We consider finite-dimensional quantum systems, i.e., systems in which a {\em state} is described by a complex matrix $\varrho$ whose eigenvalues are nonnegative and sum to one;
in symbols, $\varrho\geq 0$ and $\tr{\varrho}=1$.
If $X=\{x_1,\ldots,x_n\}$ is any finite set, a {\em measurement} with outcomes in $X$ is a collection $\A=\{\A(x_1),\ldots,\A(x_n)\}$ of $d\times d$ complex matrices such that $\A(x)\geq 0$ for all $x\in X$ and $\sum_{x\in X} \A(x) = \id$, where $\id$ denotes the identity matrix \cite{QCQI00}. The probability that the measurement $\A$ gives the outcome $x$ when performed in the state $\varrho$ is given by the Born rule $\tr{\varrho\A(x)}$. Measurements constitute a convex set, meaning that, if $\A$ and $\B$ are two measurements both having outcomes in the same set $X$, we can form their mixture $\Cc = \eta\A + (1-\eta)\B$ for any $\eta$ in the interval $[0,1]$. The mixture $\Cc$ is realized by measuring either $\A$ or $\B$ with respective probabilities $\eta$ and $1-\eta$, see Fig.~\ref{fig:mixture}. If we fix $\A$ as the reference measurement, $\Cc$ can be interpreted as a noisy version of $\A$, in which $\B$ is the noise and $\eta$ is the  
visibility 
of $\A$ in $\Cc$ \cite{DeFaKa19}.

\begin{figure}[h]
\includegraphics[scale=0.9]{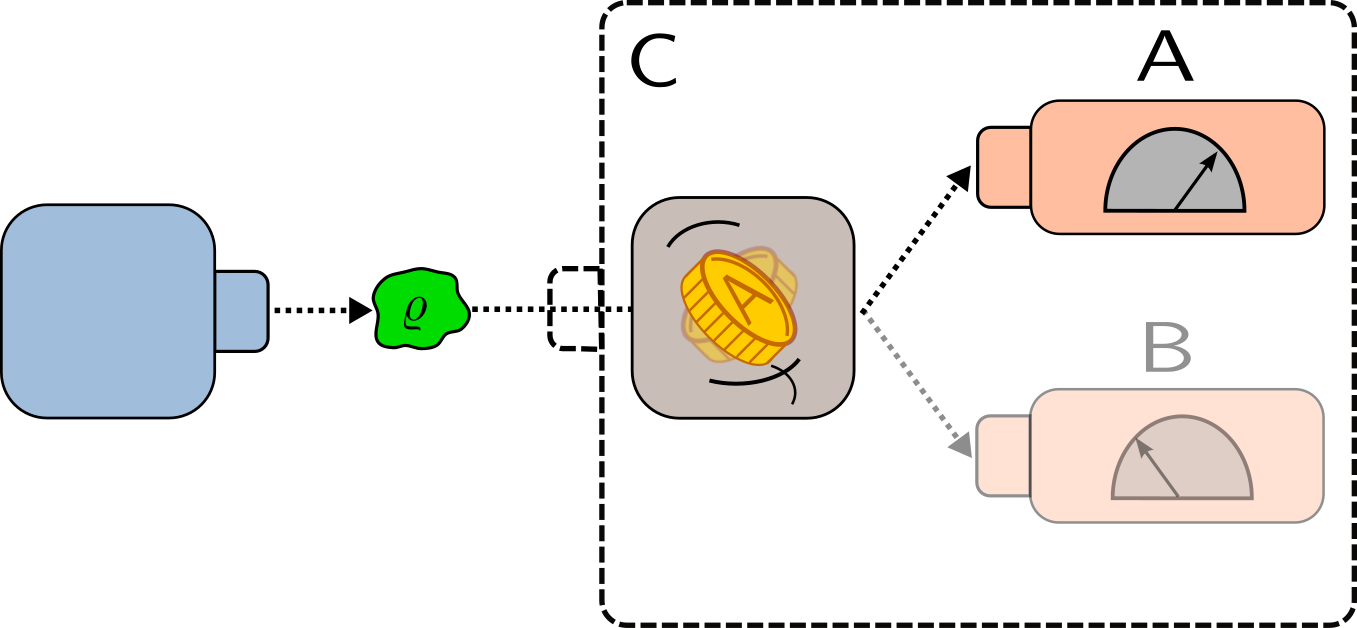}
\caption{In the mixture of two measurements, either measurement is randomly selected and then performed in the state of the system. For each measurement, the probability of being selected is the respective weight in the mixture. Only measurements sharing the same outcome set can be mixed.\label{fig:mixture}}
\end{figure}

Quantum incompatibility is a property involving two or more measurements, and it consists in the impossibility of deriving the statistics of all of them by post-processing the results of a single measurement. Formally, if $\A$ and $\B$ are measurements with respective outcome sets $X$ and $Y$, we say that $\A$ and $\B$ are {\em compatible} if they admit a {\em joint measurement} $\J$ with outcomes in the product set $X\times Y$ such that $\A(x) = \sum_{y\in Y} \J(x,y)$ and $\B(y) = \sum_{x\in X} \J(x,y)$ for all $x$ and $y$; otherwise, $\A$ and $\B$ are called {\em incompatible}. The case with more than two measurements is similar. We denote by $\Meas(X,Y)$ the set of all pairs of measurements $(\A,\B)$ with respective outcome sets $X$ and $Y$, and we write $\JMeas(X,Y)$ for the subset of those pairs in which $\A$ and $\B$ are compatible. The two sets $\Meas(X,Y)$ and $\JMeas(X,Y)$ are convex, that is, if $(\A,\B)$ and $(\Cc,\D)$ belong to either of them, then the same is true for the mixtures $(\eta\A+(1-\eta)\Cc,\eta\B+(1-\eta)\D)=: \eta\cdot(\A,\B) + (1-\eta)\cdot(\Cc,\D)$ for any any choice of $\eta\in [0,1]$.

A natural way to quantify the incompatibility content of a pair of measurements is by determining the minimal amount of noise which renders them compatible. Such an approach obviously depends on the choice of a noise model, that is, of all pairs of measurements $\ca{N}\subseteq\Meas(X,Y)$ which are regarded as noises affecting the reference measurements. Common choices for the set $\ca{N}$ are e.g.~all pairs of trivial \cite{BuHeScSt13,HeScToZi14} or uniform \cite{CaHeTo12} measurements; see \cite{DeFaKa19} for a detailed account of other possibilities. In the present paper, we fix as our noise model the full set $\ca{N}=\Meas(X,Y)$ of all pairs of measurements. The corresponding incompatibility quantifier is then the {\em incompatibility generalized robustness} $\eta^{\rm g}$, which is defined as
\begin{equation}\label{eq:def_eta}
\eta^{\rm g}(\A,\B) = \max\,\{\eta\in [0,1] \mid \eta\cdot(\A,\B) + (1-\eta)\cdot(\Cc,\D)\in\JMeas(X,Y)  \text{ for some } (\Cc,\D)\in\Meas(X,Y)\}
\end{equation}
for all pairs of measurements $(\A,\B)\in\Meas(X,Y)$. The geometric meaning of $\eta^{\rm g}(\A,\B)$ is depicted in Fig.~\ref{fig:inco_robu}; a detailed mathematical description can be found in \cite{Haapasalo15}. Here, it is enough to observe that the smaller is $\eta^{\rm g}(\A,\B)$, the more incompatible is the pair of measurements $(\A,\B)$, and that $\eta^{\rm g}(\A,\B)=1$ if and only if $(\A,\B)$ is a compatible pair. The quantity $\eta^{\rm g}(\A,\B)$ is directly related to the incompatibility robustness introduced in \cite{UoBuGuPe15}, which is defined as $\IR(\A,\B) = 1/\eta^{\rm g}(\A,\B)-1$ and decreases to $0$ as long as the pair $(\A,\B)$ becomes compatible. Our preference for using $\eta^{\rm g}(\A,\B)$ in place of the equivalent quantity $\IR(\A,\B)$ is due to the more immediate geometric interpretation.

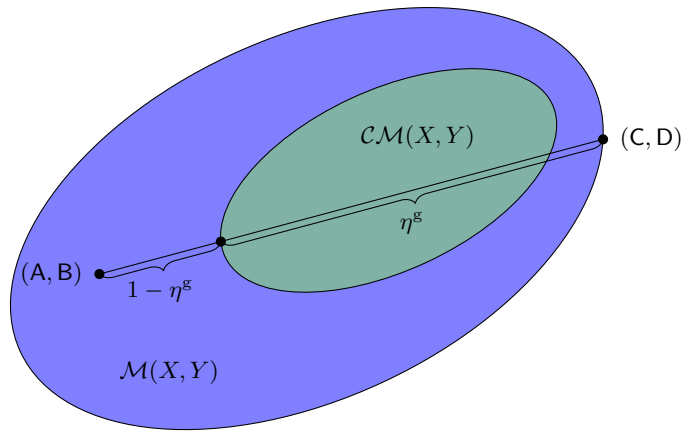
\begin{figure}[h]
\begin{tikzpicture}[scale = 1.2]

\begin{scope}[rotate=25]
\draw[name path global = ellisseA, fill = blue!50!white] (0,0) ellipse(3.5 and 2);
\draw[name path global = ellisseB, fill = darkgreen!50!white] (1,0) ellipse(2 and 1);
\path[name path global = retta] (-3.4,0.6) -- (3.4,-0.6);
\coordinate (JMeas) at (1.5,0.3);
\coordinate (Meas) at (-2.1,-0.9);
\end{scope}

\path[name intersections = {of = ellisseA and retta, by = {A,D}}];
\path[name intersections = {of = ellisseB and retta, by = {B,C}}];
\coordinate (M) at ($(A)!0.15!(D)$);

\draw[fill = black] (M)node[anchor = east, xshift = -3pt]{$(\A,\B)$} circle (0.05);
\draw[fill = black] (D)node[anchor = west, xshift = 3pt]{$(\Cc,\D)$} circle (0.05);
\draw[fill = black] (B) circle (0.05);

\draw (M) -- (D);
\draw[decorate,decoration={brace,amplitude=5pt}] (B) -- (M) node[midway, below = 5pt]{$1-\eta^{\rm g}$};
\draw[decorate,decoration={brace,amplitude=5pt}] (D) -- (B) node[midway, below = 5pt]{$\eta^{\rm g}$};

\draw (JMeas)node{$\JMeas(X,Y)$};
\draw (Meas)node{$\Meas(X,Y)$};
\end{tikzpicture}
\caption{The incompatibility generalized robustness $\eta^{\rm g}(\A,\B)$ is the maximal visibility 
of the pair of measurements $(\A,\B)$ in a compatible mixture. The noise model for $\eta^{\rm g}$ is the full set $\ca{N} = \Meas(X,Y)$ of all pairs of measurements $(\Cc,\D)$.\label{fig:inco_robu}}
\end{figure}

\subsection{State discrimination with intermediate information}\label{sec:sdw}
We now introduce the state discrimination task with intermediate information that is at the root of our strategy for experimentally assessing the incompatibility content of an arbitrary pair of measurements.

In the usual state discrimination task \cite{BaCr09,Bae13,BaKw15}, a quantum system is prepared in a state $\varrho_i$, which is randomly picked from a fixed collection of states $S=\{\varrho_1,\ldots,\varrho_n\}$. A measurement is then performed with outcomes in the set $\{1,\ldots,n\}$, and the objective is correctly guessing the value of $i$ with the outcome $i'$ given by the measurement. Clearly, in general the equality $i=i'$ can not be achieved with certainty. The best one can do is achieving it with the highest possible probability, and to this aim the measurement used for the task needs to be properly optimized.

In a state discrimination task with intermediate information \cite{GoWe10,CaHeTo18,CaHeTo21}, the above scenario is modified by adding one extra step to it. Namely, before the measurement is performed, the experimenter receives some classical partial information about the index $i$ of the unknown state, and he is then allowed to optimize the measurement accordingly. Thus, two or more different measurements can be used to detect the value of $i$, and the experimenter can switch from one to the other based on intermediate information.

In order to illustrate the latter scenario in more detail, we focus on the simplest case in which the original collection of states is partitioned in two subsets $S_a=\{\varrho_{x|a} \mid x\in X\}$ and $S_b=\{\varrho_{y|b} \mid y\in Y\}$, and intermediate information consists in knowing which subset $S_k$ the unknown state $\varrho_i$ belongs to. Here, $X$ and $Y$ are arbitrary label sets (reducing to the index set $\{1,\ldots,n\}$ in the usual scenario) and $K=\{a,b\}$ is the set containing intermediate information. We also fix two target measurements $\A$ and $\B$ with outcomes in $X$ and $Y$, respectively.
The experiment then consists in the steps below:
\begin{enumerate}[(i)]
\item choose either $k=a$ or $k=b$ with respective probabilities $q_a$ and $q_b$;
\item if $k=a$, pick the label $z\in X$ with probability $p_{z|a}$, otherwise pick the label $z\in Y$ with probability $p_{z|b}$;\label{it:step2}
\item prepare the system in the state $\varrho_{z|k}$;
\item depending on the chosen $k$, perform either the measurement $\A$ (if $k=a$) or $\B$ (if $k=b$) in the state $\varrho_{z|k}$;
\item compare the obtained measurement outcome $z'$ with the label $z$ picked in step \eqref{it:step2};
\item the experiment is successful whenever $z$ and $z'$ coincide.
\end{enumerate}
In the task described above, the success probability is 
\begin{equation}\label{eq:pre}
\PP(\en;\A,\B) = q_a \sum_{z\in X} p_{z|a}\, \tr{\varrho_{z|a}\, \A(z)} + q_b \sum_{z\in Y} p_{z|b}\, \tr{\varrho_{z|b}\, \B(z)} , 
\end{equation}
where the symbol $\en$ -- which we call a {\em partitioned state ensemble} -- collectively denotes the two sets of labeled states $S_a$ and $S_b$ together with the three probability distributions $\{p_{z|a}\mid z\in X\}$, 
$\{p_{z|b}\mid z\in Y\}$ and $\{q_k\mid k\in K\}$.

It is important to remark that the probability \eqref{eq:pre} can be empirically determined by collecting the success statistics obtained from many repetitions of the above experiment.

\subsection{Detection of incompatibility}\label{sec:doi}

In this section, we establish the connection between the incompatibility generalized robustness introduced in Section \ref{sec:inco_robu} and the state discrimination task described in Section \ref{sec:sdw}. We essentially follow \cite{SkSuCa19,UoKrShYuGu19}.

The success probability \eqref{eq:pre} easily allows to express an upper bound for the incompatibility generalized robustness \eqref{eq:def_eta}, which holds for all pairs of measurements $(\A,\B)\in\Meas(X,Y)$:
\begin{equation}\label{eq:main_0}
\eta^{\rm g}(\A,\B) \leq \frac{M(\en)}{\PP(\en;\A,\B)}\,.
\end{equation}
In the right hand side of this expression, the quantity
\begin{equation}\label{eq:Ppost}
M(\en) = \max\{\PP(\en;\Cc,\D) \mid (\Cc,\D)\in\JMeas(X,Y)\}
\end{equation}
only depends on the partitioned state ensemble and can be evaluated analytically. On the other hand, as we already pointed out, the probability $\PP(\en;\A,\B)$ is experimentally assessable by collecting the success statistics of many runs of the experiment.

The essential point proved in \cite{SkSuCa19,UoKrShYuGu19}, however, is that there always exists a particular choice of the partitioned state ensemble $\en$ for which the inequality \eqref{eq:main_0} actually becomes an equality. Such a choice, of course, depends on the measurements $\A$ and $\B$. We thus see that, if the preparation is suitably tuned, then the incompatibility generalized robustness can be determined in a state discrimination task with intermediate information.


\subsection{Unbiased qubit measurements}\label{sec:coq}
In our experiment, we focus on a qubit system $\hh=\Cb^2$ and we want to determine the incompatibility of two dichotomic measurements $\A$ and $\B$. For simplicity, we assume that the outcome sets of $\A$ and $\B$ coincide and $X=Y=\{+,-\}$.

As a remarkable fact, for a particular choice of the partitioned state ensemble $\en$, the inequality \eqref{eq:main_0} becomes an equality for a quite large class of measurements. Namely, for all $z\in\{+,-\}$ and $k\in\{a,b\}$, we set
\begin{equation}\label{eq:ensex}
p_{z|k} = q_k = \tfrac{1}{2}\,,\qquad\qquad \varrho_{z|k} = \tfrac{1}{2}\big(\id+z\,\vu_k\cdot\vsigma\big)
\end{equation}
with
\begin{equation}\label{eq:en_qubit2}
\vu_a=\tfrac{1}{\sqrt{2}}\big(\ve_1 + \ve_2\big)\,, \qquad\qquad \vu_b=\tfrac{1}{\sqrt{2}}\big(\ve_1 - \ve_2\big)\,.
\end{equation}
In these expressions, $\ve_1$, $\ve_2$ and $\ve_3$ are the unit vectors along the coordinate axes of $\Rb^3$, and for any vector $\vv = v_1\ve_1 + v_2\ve_2 + v_3\ve_3$ we set $\vv\cdot\vsigma = v_1\sigma_1 + v_2\sigma_2 + v_3\sigma_3$, where $\sigma_1$, $\sigma_2$ and $\sigma_3$ are the three Pauli matrices. For two fixed directions $\va$ and $\vb$ and real numbers $0<\alpha,\beta\leq 1$, we consider measurements of the form
\begin{equation}\label{eq:AB_qubit1}
\A_{\alpha\va}(z) = \tfrac{1}{2}\big(\id + z\,\alpha\,\va\cdot\vsigma\big)\,, \qquad\qquad\B_{\beta\vb}(z) = \tfrac{1}{2}\big(\id + z\,\beta\,\vb\cdot\vsigma\big)\,.
\end{equation}
These measurements constitute two noisy unbiased qubit measurements along $\va$ and $\vb$ with possibly different noise parameters $\alpha$ and $\beta$. The smaller are the parameters $\alpha$ and $\beta$, the more intense is the noise in the respective measurements. In the noiseless case $\alpha=\beta=1$, the measurements $\A_{\alpha\va}$ and $\B_{\beta\vb}$ are projective.

\begin{figure}[h]
\includegraphics[width=.5\columnwidth]{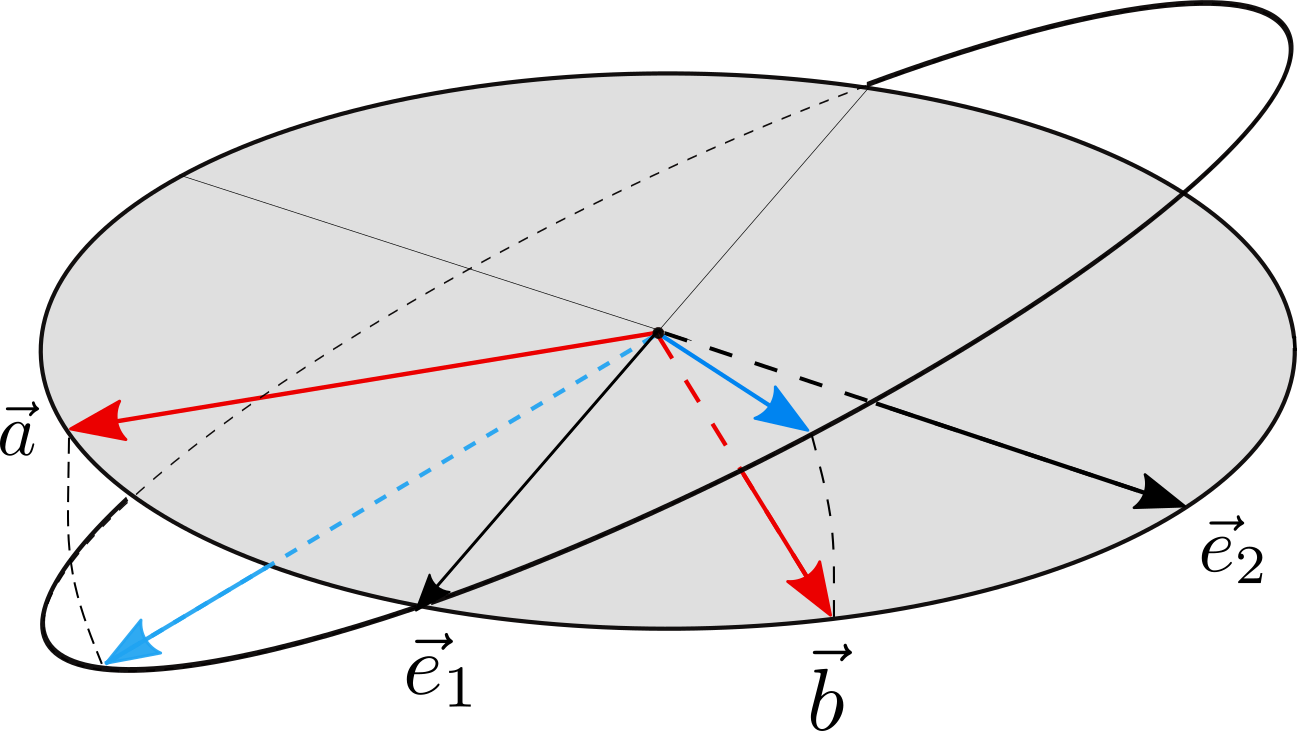}

\caption{Geometric representation of the directions identifying the measurements \eqref{eq:AB_qubit1} considered in the experiment. The red arrows represent the vectors $\va$ and $\vb$ defined in \eqref{eq:AB_qubit2}, which are symmetric with respect to $\ve_1$ and span the angle $2\theta$. The blue arrows represent the vectors \eqref{eq:rotated_ab} obtained by tilting the previous vectors $\va$ and $\vb$ by the angle $\phi$ around $\ve_1$.} \label{fig:axis}
\end{figure}

For the partitioned state ensemble defined in \eqref{eq:ensex}-\eqref{eq:en_qubit2}, the maximum \eqref{eq:Ppost} was evaluated in \cite{GoWe10,CaHeTo18} and found to be
\begin{equation}
M(\en) = \frac{1}{2}\left(1+\frac{1}{\sqrt{2}}\right) \,,
\end{equation}
so that the quantity on the right hand side of \eqref{eq:main_0} reduces to
\begin{equation}\label{eq:chiab}
\chi(\en;\A,\B) = \frac{\sqrt{2}+1}{2\sqrt{2}\,\PP(\en;\A,\B)}\,.
\end{equation}
Remarkably, it turns out that the bound \eqref{eq:main_0} is saturated by the measurements \eqref{eq:AB_qubit1} when $\va$ and $\vb$ lie in the plane spanned by $\ve_1$ and $\ve_2$, they are symmetric around $\ve_1$, and we further enforce the same noise parameter in the two measurements (see Fig.~\ref{fig:axis}). More precisely, for
\begin{equation}\label{eq:AB_qubit2}
\va = \cos\theta\,\ve_1 + \sin\theta\,\ve_2\,,\qquad\qquad\vb = \cos\theta\,\ve_1 - \sin\theta\,\ve_2\,,\qquad\qquad 0<\theta\leq\frac{\pi}{2}\,,
\end{equation}
we have
\begin{equation}\label{eq:main_1}
\eta^{\rm g}(\A_{\gamma\va},\B_{\gamma\vb}) = \chi(\en;\A_{\gamma\va},\B_{\gamma\vb})
\end{equation}
provided that the noise parameter $\gamma$ satisfies
\begin{equation}\label{eq:gamma}
\frac{1}{\cos\theta+\sin\theta} < \gamma \leq 1
\end{equation}
(the proof can be found in \cite{DeFaKa19} for the case $\gamma=1$ and in the appendix for $\gamma<1$). For the other values of $\va$, $\vb$, $\alpha$ and $\beta$, the inequality \eqref{eq:main_0} still holds true and yields
\begin{equation}\label{eq:main_2}
\eta^{\rm g}(\A_{\alpha\va},\B_{\beta\vb}) \leq \chi(\en;\A_{\alpha\va},\B_{\beta\vb})  \,.
\end{equation}
In particular, the inequality is trivial when $\A_{\alpha\va}$ and $\B_{\beta\vb}$ are compatible, since in this case $\eta^{\rm g}(\A_{\alpha\va},\B_{\beta\vb}) = 1$ and $\PP(\en;\A_{\alpha\va},\B_{\beta\vb}) \leq M(\en)$ by definition. This happens e.g.~when 
\begin{equation}\label{eq:triv}
\alpha = \beta\leq \frac{1}{\cos\theta+\sin\theta},
\end{equation}
which motivates the constraint on $\gamma$ given in \eqref{eq:gamma}. 

Summarizing, the success probability in the partitioned state ensemble \eqref{eq:ensex}-\eqref{eq:en_qubit2} is the quantity we are going to estimate through our experiments, as it allows to evaluate the incompatibility of the measurements $\A_{\alpha\va}$ and $\B_{\beta\vb}$ via \eqref{eq:chiab} and \eqref{eq:main_1}, or at least upper bound it via \eqref{eq:main_2}, depending on the values of the noise parameters $\alpha$ and $\beta$ and on the directions $\vec{a}$ and $\vec{b}$.

\section{Experimental procedure and results}\label{sec:the}

\subsection{Experimental apparatus}\label{sec:ea}
\begin{figure}
  \includegraphics[width=.9\columnwidth]{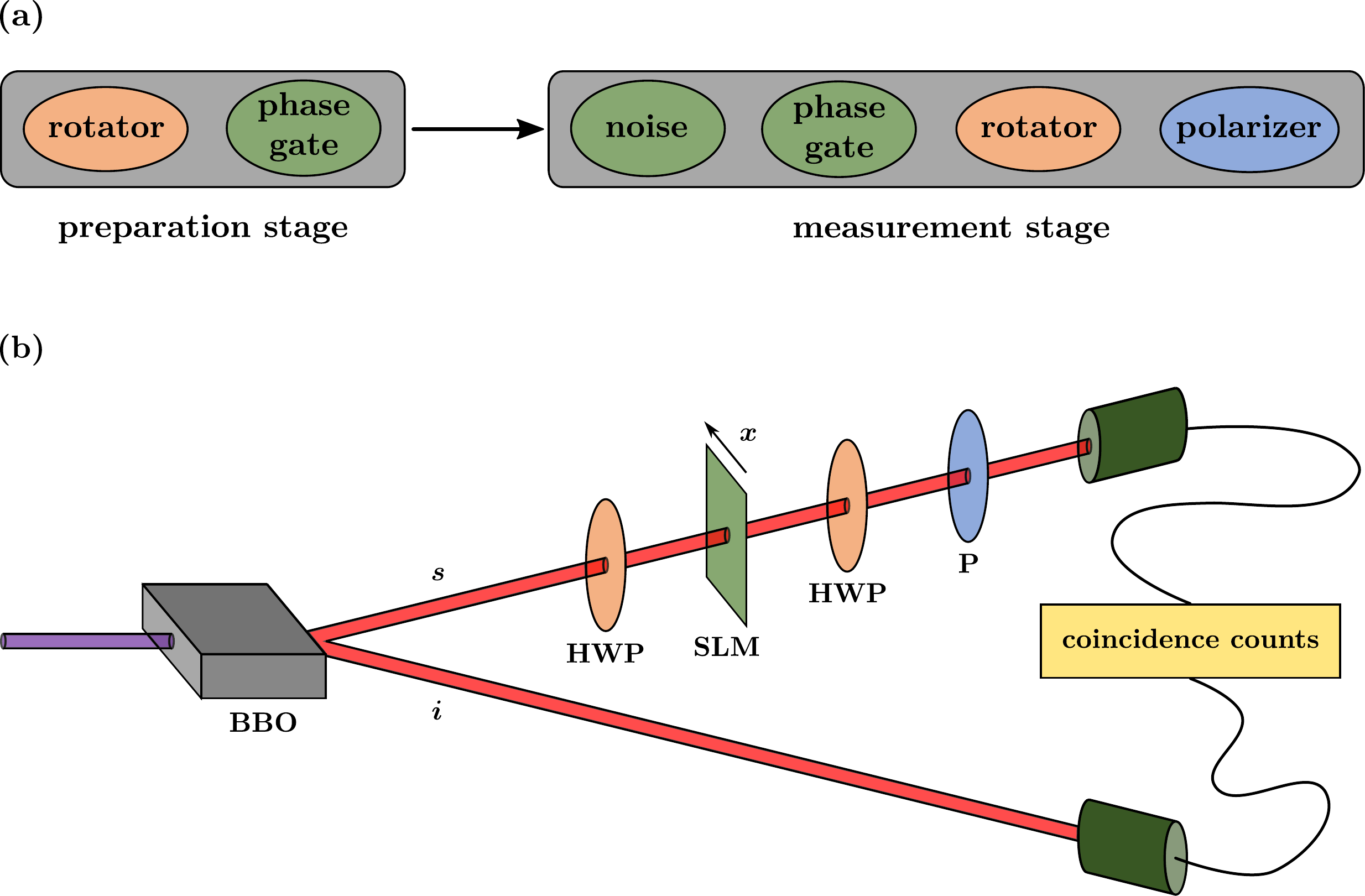}
  \caption{\textbf{(a)} Logical scheme of the experiment. The system goes through state preparation and measurement stages. In the preparation stage, a rotation and a phase gate operation generate a pure polarization state, starting from a horizontally polarized photon. Analogously, in the measurement stage, a phase combined with another rotation and the action of a polarizer on the horizontal direction realize a polarization projective measurement. Possibly, the presence of phase diffusion leads to noisy measurements (see \eqref{eq:nnoisy}). \textbf{(b)} Sketch of the experimental apparatus. The signal photon (\textit{s}) exiting the non-linear crystal (BBO) is first rotated by means of a half-wave plate (HWP). The phase-gate operations as well as the introduction of noise are then realized via a spatial light modulator (SLM) affecting the photon phase with control in the $x$ direction. A second half-wave plate (HWP) and polarizer (P) fully determine the measurement.}\label{fig:app}
\end{figure}

The logical scheme of our experiment is illustrated in Fig.~\ref{fig:app}(a).
The two main parts are the state preparation stage,
where the input states are generated by a polarization rotator and a phase gate, 
and the measurement stage,
which consists of a phase gate, a phase diffusion, a
polarization rotator and finally a projection element realized by a polarizer.
By using a polarization rotator and a phase gate, it is possible to generate a pure state of arbitrary polarization starting from a single photon with horizontal polarization. This allows us to obtain
the 4 states \eqref{eq:ensex}-\eqref{eq:en_qubit2} with a high degree of purity (as shown in the next section), as well as to realize general projective measurements, in particular the
measurements \eqref{eq:AB_qubit1} for $\alpha=\beta=1$. 
Denoting by $\ket{H}$ and $\ket{V}$ the horizontal and vertical polarization vectors and setting
\begin{equation}\label{eq:ppure}
\ket{\psi_{z|a} (\theta)} = \frac{1}{\sqrt{2}}\left(\ket{H}+ z \rme^{i \theta}\ket{V}\right)\,,\qquad\qquad \ket{\psi_{z|b} (\theta)} = \frac{1}{\sqrt{2}}\left(\ket{H}+ z \rme^{-i \theta}\ket{V}\right)
\end{equation}
for $z\in\{+,-\}$, the 4 pure states \eqref{eq:ensex}-\eqref{eq:en_qubit2} are realized as
\begin{equation}\label{eq:equiv}
\varrho_{z|a} = \kb{\psi_{z|a}(\pi/4)}{\psi_{z|a}(\pi/4)} \,, \qquad\qquad \varrho_{z|b} = \kb{\psi_{z|b}(\pi/4)}{\psi_{z|b}(\pi/4)} \,.
\end{equation}
Concerning the measurements \eqref{eq:AB_qubit1} with $\alpha=\beta=1$, 
we first choose $\va$ and $\vb$ as in \eqref{eq:AB_qubit2},
which corresponds to $\A_\va$ and $\B_\vb$ being the projective measurements in the bases $\left\{\ket{\psi_{+|a}(\theta)}, \ket{\psi_{-|a}(\theta)}\right\}$ and $\left\{\ket{\psi_{+|b}(\theta)}, \ket{\psi_{-|b}(\theta)}\right\}$, respectively. 
In addition, we also consider the case where $\va$ and $\vb$ are rotated around $\vec{e}_1$ by an angle $\phi$, while the angle between them remains fixed and equal to $2\theta=\pi/2$. In the latter case,
\begin{equation}\label{eq:rotated_ab}
\va = \frac{1}{\sqrt{2}} \left(\ve_1 + \cos\phi\,\ve_2 + \sin\phi\,\ve_3\right)\,,\qquad\qquad \vb = \frac{1}{\sqrt{2}} \left(\ve_1 - \cos\phi\,\ve_2 - \sin\phi\,\ve_3\right)\,,
\end{equation}
and $\A_\va$ and $\B_\vb$ are the projective measurements in the bases $\left\{\ket{\varphi_{+|a}(\phi)}, \ket{\varphi_{-|a}(\phi)}\right\}$ and $\left\{\ket{\varphi_{+|b}(\phi)}, \ket{\varphi_{-|b}(\phi)}\right\}$ respectively, where
\begin{equation}\label{eq:tiltm}
\ket{\varphi_{z|a}(\phi)} = \sqrt{\eta_z(\phi)}\ket{H} + z\rme^{i \varsigma_z(\phi)}\sqrt{1-\eta_z(\phi)} \ket{V}\,,\qquad\qquad \ket{\varphi_{z|b}(\phi)} = \sqrt{1-\eta_z(\phi)}\ket{H} + z\rme^{-i \varsigma_z(\phi)}\sqrt{\eta_z(\phi)} \ket{V}
\end{equation}
with
\begin{equation}\label{eq:cases}
\eta_z(\phi) = \frac{1}{2}\left( 1+ \frac{z\sin\phi}{\sqrt{2}}\right)\,,\qquad\qquad \varsigma_z(\phi)=\arctan\cos\phi\,.
\end{equation}
In the experiment, the phases $\varsigma_z(\phi)$ and the
amplitudes $\eta_z(\phi)$ are set by the phase gate and
rotator in the projection stage. 
Finally, it is possible to implement the noisy measurements $\A_{\alpha\va}$ and $\B_{\beta\vb}$ by adding a phase diffusion. This results in a mixture of the measurements $\A_\va$ and $\B_\vb$ with the uniform trivial measurement, i.e., 
the measurement that outputs $+$ or $-$ with equal probabilities and independently of the system state:
\begin{equation}\label{eq:nnoisy}
\A_{\alpha\va} = \alpha\A_\va + (1-\alpha) \tfrac 12 \id\,,\qquad\qquad \B_{\beta\vb} = \beta\B_\vb + (1-\beta) \tfrac 12 \id\,.
\end{equation}

The actual experimental implementation of the scheme described above is depicted in Fig.~\ref{fig:app}(b). Twin photons are generated via parametric down-conversion by a $1\,{\rm mm}$ thick beta-Barium Borate non-linear crystal (BBO), which is pumped with a radiation at $405\,{\rm nm}$ generated by a laser diode. One of the two photons (\textit{i}=idler) is merely detected.
The other photon (\textit{s}=signal) is used to implement the procedure illustrated in Fig.~\ref{fig:app}(a).  
First, its polarization is rotated by the proper angle by means of a half-wave plate (HWP). Then, the phase gate is implemented by a spatial light modulator (SLM), 
which is a 1D liquid crystal mask (640 pixel, $100\,\mu {\rm m}$\,/\,pixel).
The SLM is controlled by computer along the $x$ direction and it is used to introduce the desired phase for each pixel.
By means of a single SLM, it is possible to implement both the phase gate operation in the preparation stage and the phase gate as well as
the phase diffusion in the projection stage. 
The two phase gate operations are performed by setting all the pixels to the same proper value in the signal part of the mask.
On the other hand, in order to implement the phase diffusion, we add a random 
phase that is different for each pixel. 
In particular, we use a random function consisting in a flat distribution between the two values 
$-\ell/2$ and $\ell/2$, so that the
mixing factors $\alpha$ and $\beta$ of the noisy measurements \eqref{eq:nnoisy} are functions of $\ell$.
Finally, we use a second HWP and a polarizer (P) with axis on the horizontal direction to complete
the projection stage. In the first part of the experiment, these last two devices allow us to select among the projective measurements in the bases $\left\{\ket{\psi_{+|a}(\theta)}, \ket{\psi_{-|a}(\theta)}\right\}$ and $\left\{\ket{\psi_{+|b}(\theta)}, \ket{\psi_{-|b}(\theta)}\right\}$, while in the second part we use them to set the proper values of  $\varsigma_z(\phi)$ and 
$\eta_z(\phi)$ in (\ref{eq:cases}) and to select among
$\left\{\ket{\varphi_{+|a}(\phi)}, \ket{\varphi_{-|a}(\phi)}\right\}$ and $\left\{\ket{\varphi_{+|b}(\phi)}, \ket{\varphi_{-|b}(\phi)}\right\}$.
The two HWPs are rotated by step motors controlled by computer. 
The coincidence counts are detected by two home-made single photon detectors.


\subsection{Measurement of the success probability}\label{sec:mot}

\begin{figure}[t]
 \includegraphics[width=0.45\columnwidth]{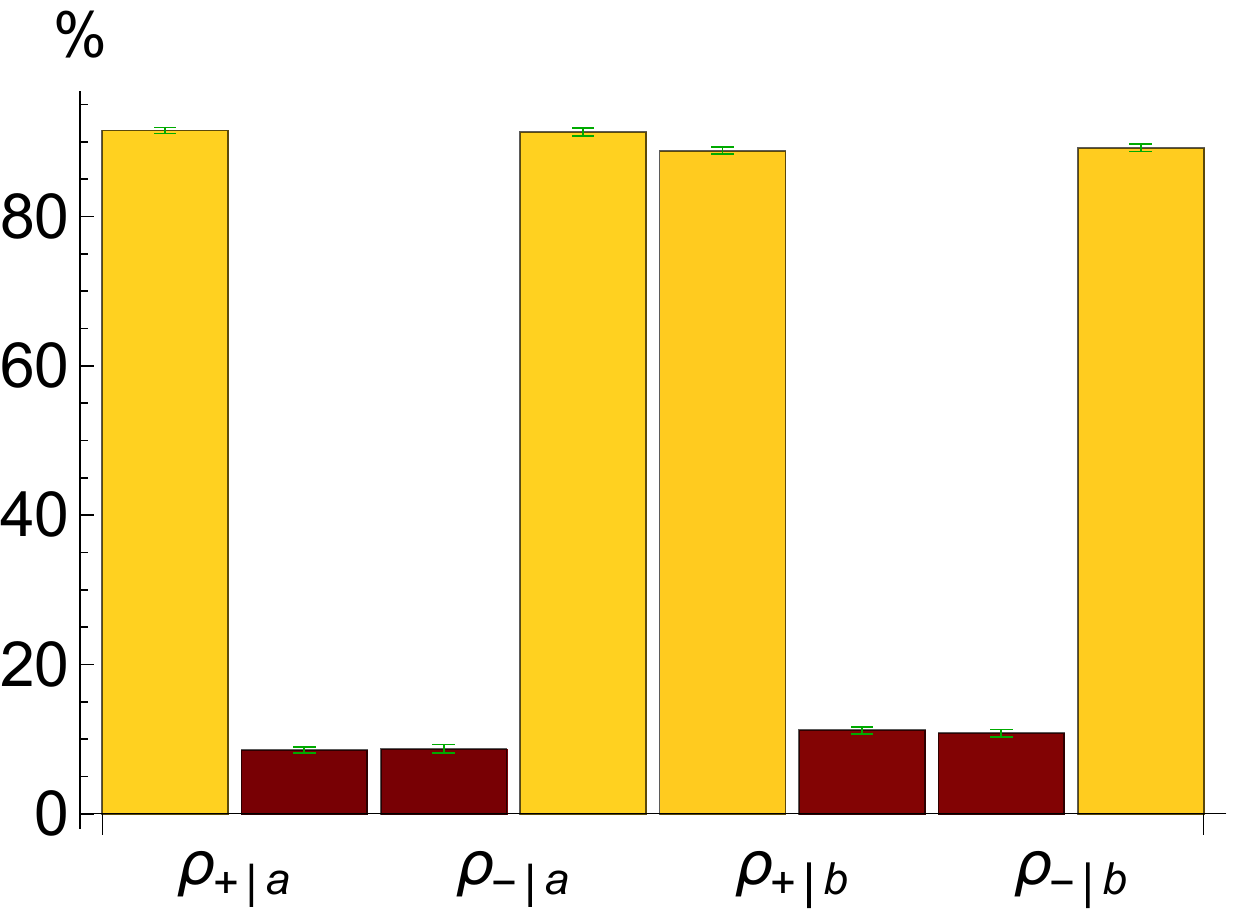}
\hspace{0.3cm}\includegraphics[width=0.45\columnwidth]{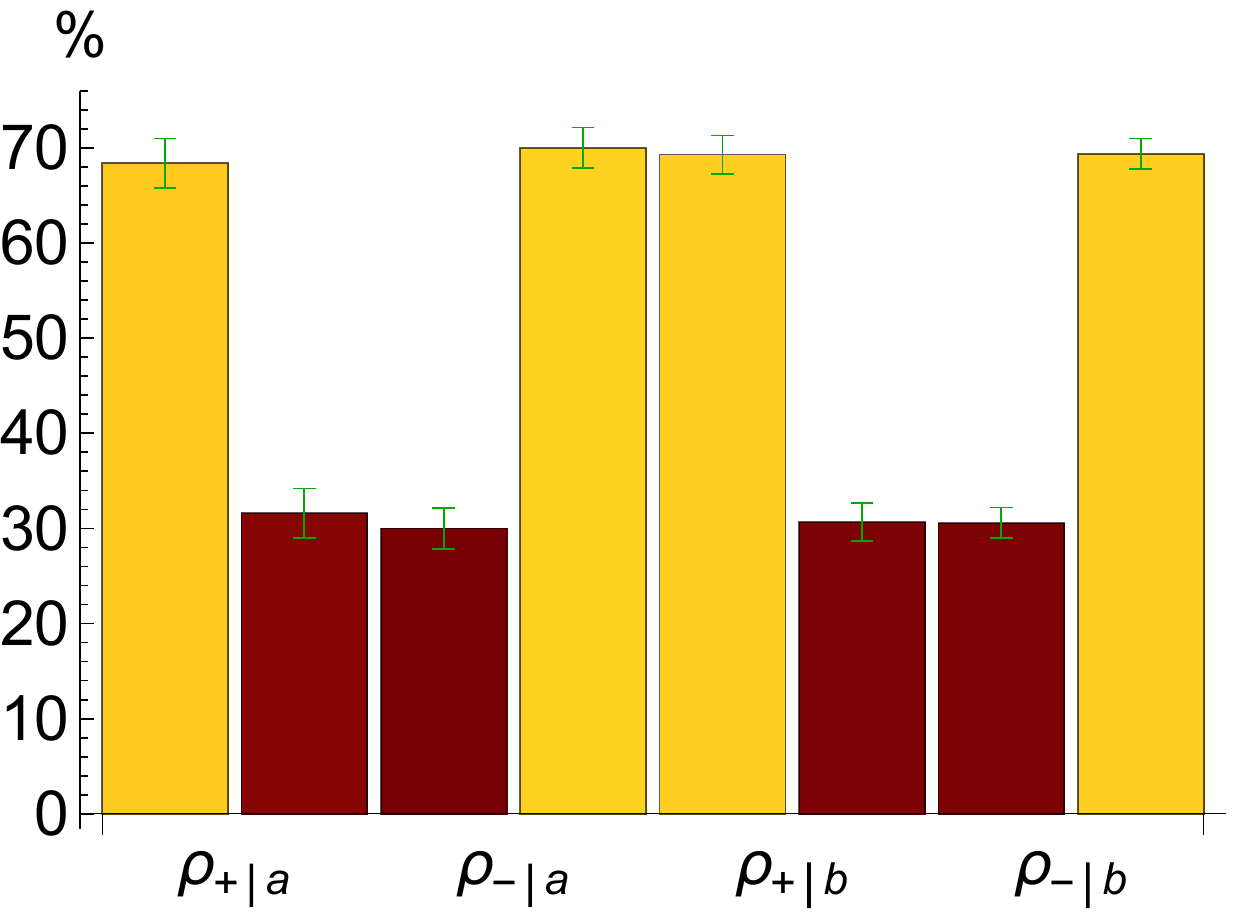}
\caption{Percentage of photon counts for the 4 prepared states 
$\varrho_{+|a},\varrho_{-|a},\varrho_{+|b},\varrho_{-|b}$ in \eqref{eq:ensex}-\eqref{eq:en_qubit2} and \eqref{eq:equiv}. For each state, the two bars are the percentages
associated with the two possible outcomes (respectively, $+$ and $-$)
of the corresponding measurement, which is $\A_{\alpha\va}$
for the states $\varrho_{+|a}$, $\varrho_{-|a}$, and 
$\B_{\beta\vb}$ for the states $\varrho_{+|b}$, $\varrho_{-|b}$. The vectors $\va$ and $\vb$ are given by \eqref{eq:AB_qubit2} with $\theta=10^\circ$. {\textbf{Left:}} projective measurements with $\alpha=\beta=1$. {\textbf{Right:}} noisy measurements with $\alpha=\beta=0.6$.
By \eqref{eq:pre}, the success probability $P(\en;\A_{\alpha\va}, \B_{\beta\vb})$ is the equally weighted convex sum of the 4 empirical frequencies in dark red color.}\label{fig:hist} 
\end{figure}

As explained in Section \ref{sec:coq}, the incompatibility generalized robustness of the two qubit measurements $\A_{\alpha\va}$ and $\B_{\beta\vb}$ can be assessed or at least upper bounded by knowing the value of the success probability $\PP(\en;\A_{\alpha\va},\B_{\beta\vb})$ for the partitioned state ensemble \eqref{eq:ensex}-\eqref{eq:en_qubit2}.
Here, we recall that $\PP(\en;\A_{\alpha\va},\B_{\beta\vb})$ is obtained either by performing the measurement $\A_{\alpha\va}$ in a state randomly chosen with uniform probability among the two states $\varrho_{+|a}, \varrho_{-|a}$ of the first subensemble, or by measuring $\B_{\beta\vb}$ in one of the two states of the second subensemble, chosen again with uniform probability among $\varrho_{+|b}, \varrho_{-|b}$. The two alternative preparation and measurement procedures are also equally probable. We get a success each time the outcome $z'$ of 
$\A_{\alpha\va}$ (respectively, $\B_{\beta\vb}$) coincides with the label $z$ of the prepared state $\varrho_{z|a}$ (resp., $\varrho_{z|b}$). The resulting probability $\PP(\en;\A_{\alpha\va},\B_{\beta\vb})$ is then half the sum of the occurrence probabilities of the latter coincidences.

In our setup, the success probabilities of the two alternative procedures are determined by the photon counts in the presence of the respective states and measurements.
In Fig.~\ref{fig:hist}, we report the photon counts along with their uncertainties for two exemplary pairs of measurements, namely, the two projective measurements $\A_\va$ and $\B_\vb$ with $\va$ and $\vb$ given by \eqref{eq:AB_qubit2} for $\theta = 10^\circ$, and their uniformly noisy versions \eqref{eq:nnoisy} with equal noise parameters $\alpha=\beta=0.6$.
Among the 8 resulting frequencies (two possible outcomes for each initial state), only 4 ones contribute to the definition of $\PP(\en;\A_{\alpha\va},\B_{\beta\vb})$.
We note that, in the case of projective measurements, the uncertainties are dominated
by the statistical fluctuations associated with the number of photons, since these fluctuations are very close to the shot noise distribution.
From our data, we can also deduce a purity of 0.985.
On the other hand, in the presence of a noisy measurement, we have a further contribution
that is due to the uncertainty introduced by
the assessment of the actual value of the equal noise $\gamma = \alpha = \beta$. As mentioned above, such noise is implemented via random
phases introduced by the SLM, whose discrete grid of pixels induces an
uncertainty on the value of the parameter $\gamma$
that increases as the noise becomes larger. This leads to deviations from the shot-noise distribution that are progressively more significant for larger values of the noise.

\subsection{Incompatibility generalized robustness for equally noisy measurements}
\begin{figure}[t]
 \includegraphics[width=0.45\columnwidth]{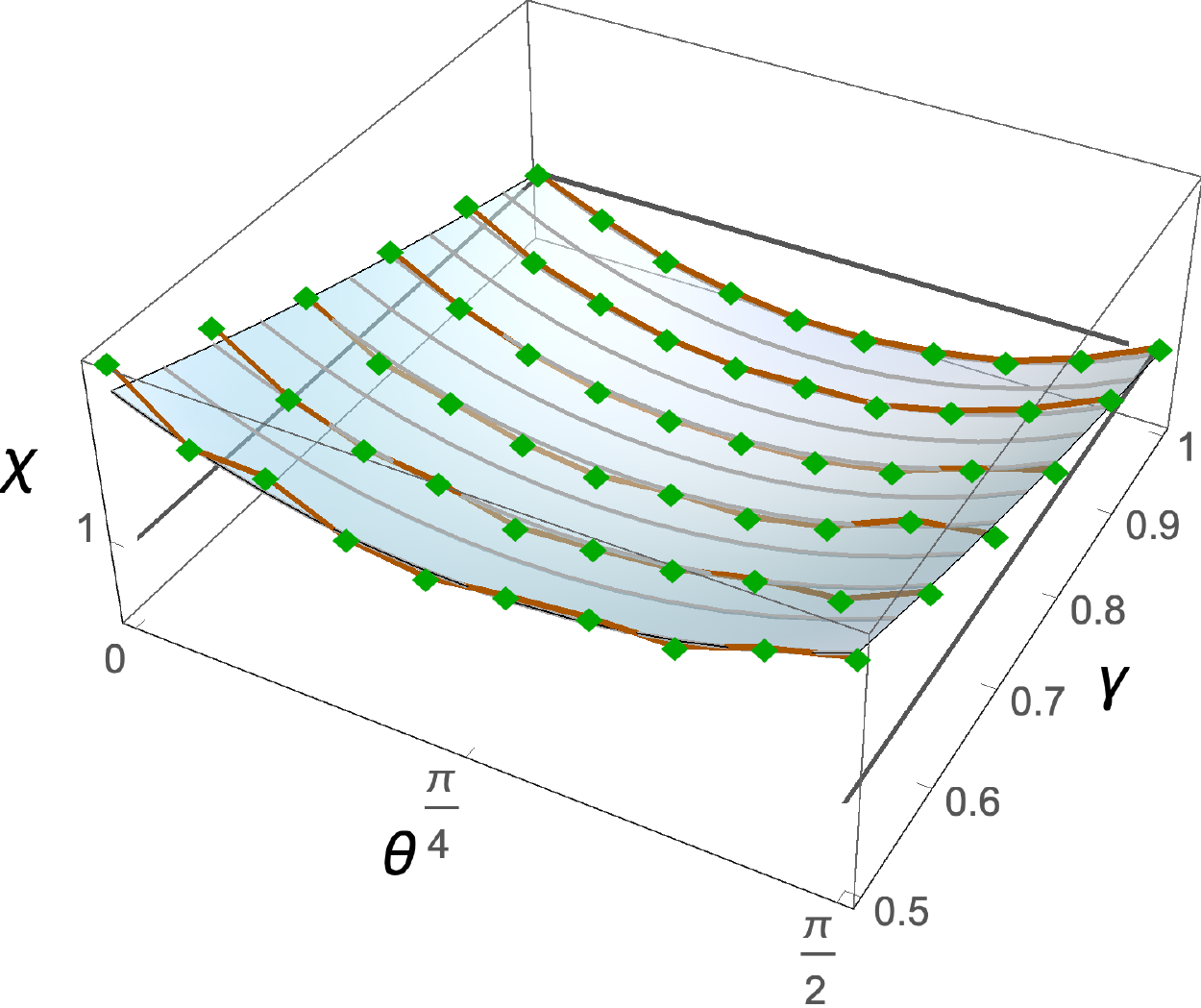}
\textbf{(a)}
 \includegraphics[width=0.45\columnwidth]{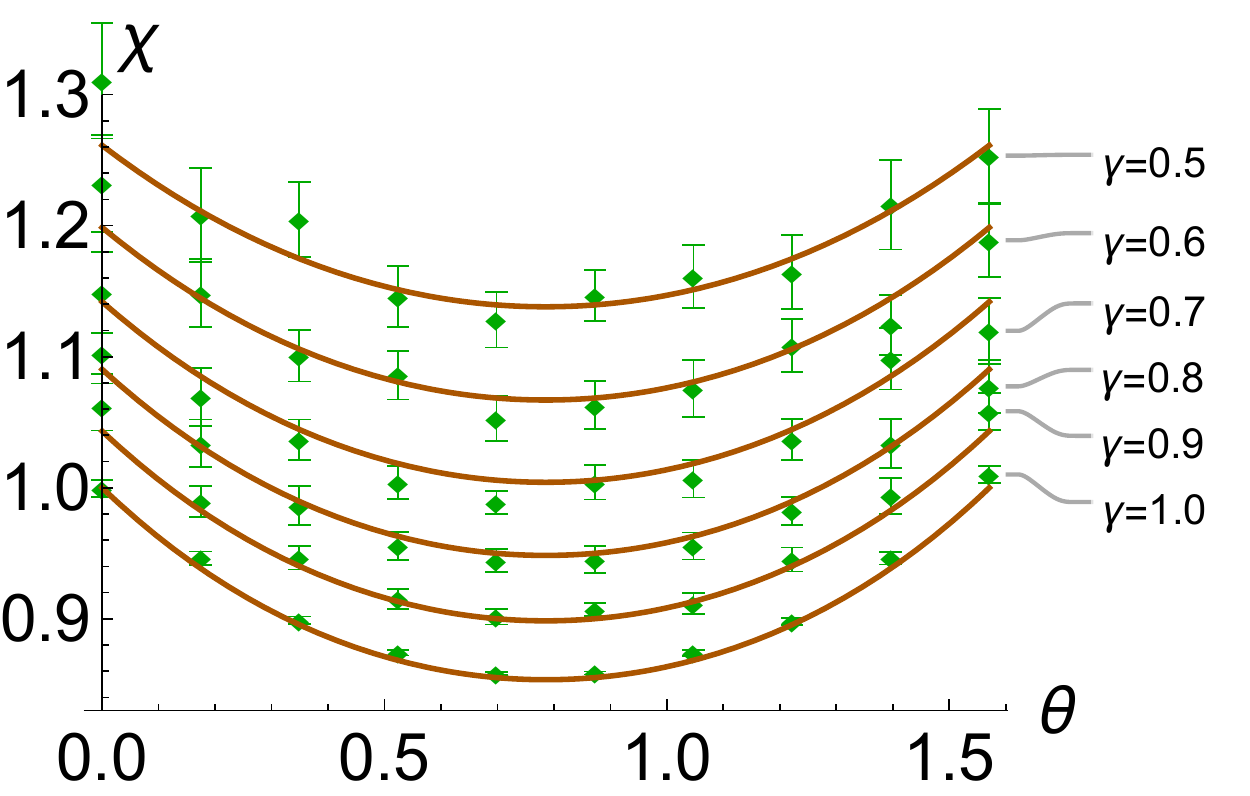}
\textbf{(b)}
\caption{The empirically assessable quantity $\chi(\en;\A_{\gamma\va},\B_{\gamma\vb})$ defined in \eqref{eq:chiab} and evaluated for the measurements $\A_{\gamma\va}$ and $\B_{\gamma\vb}$ with $\va = \cos\theta\,\ve_1 + \sin\theta\,\ve_2$ and $\vb = \cos\theta\,\ve_1 - \sin\theta\,\ve_2$. These measurements are obtained by mixing the uniform trivial measurement with the projective measurements $\A_\va$ and $\B_\vb$ in the respective bases $\left\{\ket{\psi_{+|a}(\theta)}, \ket{\psi_{-|a}(\theta)}\right\}$ and $\left\{\ket{\psi_{+|b}(\theta)}, \ket{\psi_{-|b}(\theta)}\right\}$, see \eqref{eq:ppure}, \eqref{eq:nnoisy}.
As discussed in Sec.~\ref{sec:coq}, the quantity $\chi(\en;\A_{\gamma\va},\B_{\gamma\vb})$ coincides with the incompatibility generalized robustness $\eta^{\rm g}(\A_{\gamma\va},\B_{\gamma\vb})$ 
whenever its value lies below $1$, while $\chi(\en;\A_{\gamma\va},\B_{\gamma\vb}) \geq 1$ implies that the measurements $\A_{\gamma\va}$ and $\B_{\gamma\vb}$ are compatible. 
{\textbf{(a)}} The quantity $\chi(\en;\A_{\gamma\va},\B_{\gamma\vb})$ as a function of $\gamma$ and $\theta$. The semi-transparent surface is the theoretical prediction \eqref{eq:tpr}, while the lozenges represent the experimental data. The solid lines connecting the lozenges at fixed values of $\gamma$ only serve as a guide to the reader's eye. {\textbf{(b)}} Sections of the plot on the left for different values of $\gamma$. Here, the experimental data are reported with their error bars and the solid lines represent the theoretical predictions.}\label{fig:equal}
\end{figure}


After experimentally determining the success probability $\PP(\en;\A_{\alpha\va},\B_{\beta\vb})$, we can evaluate the quantity $\chi(\en;\A_{\alpha\va},\B_{\beta\vb})$ defined in \eqref{eq:chiab}. We begin with the case of equally noisy measurements, for which we fix $\alpha=\beta = \gamma$. Moreover, we first assume that the unit vectors $\va$ and $\vb$ are spanned by $\ve_1$ and $\ve_2$ and are symmetric around $\ve_1$ as in \eqref{eq:AB_qubit2} (see Fig.~\ref{fig:axis}). In this case, by the discussion of Sec.~\ref{sec:coq}, the incompatibility generalized robustness $\eta^{\rm g}(\A_{\gamma\va},\B_{\gamma\vb})$ coincides with $\chi(\en;\A_{\gamma\va},\B_{\gamma\vb})$ for all values of the noise parameter $\gamma$ within the range \eqref{eq:gamma}. The latter values are exacty those that render the two measurements $\A_{\gamma\va}$ and $\B_{\gamma\vb}$ incompatible. On the other hand, for the partitioned state ensemble \eqref{eq:ensex}-\eqref{eq:en_qubit2}, the quantity $\chi(\en;\A_{\gamma\va},\B_{\gamma\vb})$ can be easily evaluated explicitly, resulting in
\begin{equation}\label{eq:tpr}
\chi(\en;\A_{\gamma\va},\B_{\gamma\vb}) = \frac{\sqrt{2}+1}{\sqrt{2}+\gamma\,(\cos\theta+\sin\theta)}\,. 
\end{equation}
Therefore, by comparing this formula with the experimental data, we empirically verify the value of $\eta^{\rm g}(\A_{\gamma\va},\B_{\gamma\vb})$ predicted by the theory when $\A_{\gamma\va}$ and $\B_{\gamma\vb}$ are incompatible. For the values of $\gamma$ below the interval \eqref{eq:gamma}, the two measurements are compatible, hence $\eta^{\rm g}(\A_{\gamma\va},\B_{\gamma\vb})=1$ and accordingly we must find $\chi(\en;\A_{\gamma\va},\B_{\gamma\vb})\geq 1$.    

The experimental results are shown in Fig.~\ref{fig:equal} and confirm the theoretical predictions. In Fig.~\ref{fig:equal}(a) the semi-transparent surface is the theoretical value \eqref{eq:tpr} of $\chi(\en;\A_{\gamma\va},\B_{\gamma\vb})$, plotted as a function of $\gamma$ and $\theta$. The lozenges represent the results of the experiment, obtained from the photon counts described in Sec.~\ref{sec:mot}. 
In Fig.~\ref{fig:equal}(b) the 2D graph displays $\chi(\en;\A_{\gamma\va},\B_{\gamma\vb})$ as a function of $\theta$ for several fixed values of $\gamma$, and it provides a clear-cut comparison with the experimental points and their error bars. We observe the agreement between the experimental data and the theoretical predictions
for all considered range of parameters. As it clearly appears,  the experimental uncertainties increase for small values of the parameter $\gamma$, i.e., for more intense noise in the measurements, because the fluctuations of $\gamma$ become larger due to the finite pixel resolution of the
SLM.
We already remarked that in the present situation the measurements $\A_{\gamma\va}$ and $\B_{\gamma\vb}$ are compatible exactly when $\chi(\en;\A_{\gamma\va},\B_{\gamma\vb})\geq 1$. In the noiseless case with $\gamma=1$, this happens only 
at the extreme values $\theta=0$ and $\theta=\pi/2$, where the projective measurements $\A_\va$ and $\B_\vb$ coincide (for $\theta=0$) or coincide up to swapping their outcomes (for $\theta=\pi/2$).
For values of $\gamma$ smaller than $1$, the range of angles $\theta$ for which $\A_{\gamma\va}$ and $\B_{\gamma\vb}$ are compatible progressively enlarges. On the other hand, for $\chi(\en;\A_{\gamma\va},\B_{\gamma\vb})<1$, the two measurements are incompatible and $\chi(\en;\A_{\gamma\va},\B_{\gamma\vb})$
coincides with $\eta^{\rm g}(\A_{\gamma\va},\B_{\gamma\vb})$.
Independently of the value of $\gamma$, the incompatibility generalized robustness has a minimum for $\theta=\pi/4$,
where the measurements are then maximally incompatible. 
The most incompatible measurements are indeed the projective measurements with $\theta=\pi/4$, 
for which $\eta^{\rm g}(\A_{\vu_a},\B_{\vu_b})=(2+\sqrt{2})/4\approx 0.85$. Regardless of $\theta$, the incompatibility of $\A_{\gamma\va}$ and $\B_{\gamma\vb}$ reduces by decreasing the noise parameter $\gamma$.

\begin{figure}[t]
 \includegraphics[width=0.45\columnwidth]{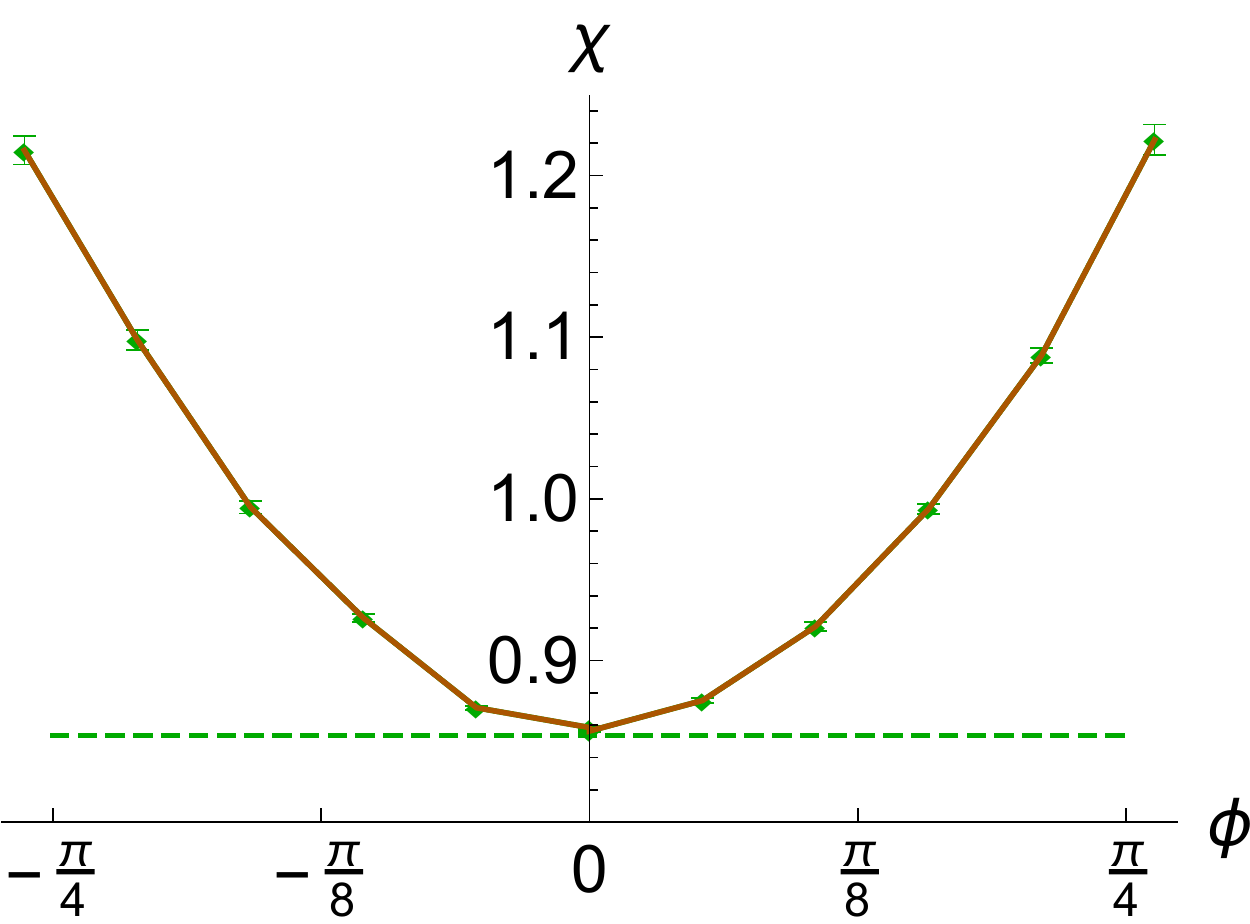}
\caption{The quantity $\chi(\en;\A_\va,\B_\vb)$ plotted as a function of $\phi$ for the projective measurements $\A_\va$ and $\B_\vb$ in the respective bases $\left\{\ket{\varphi_{+|a}(\phi)}, \ket{\varphi_{-|a}(\phi)}\right\}$ and $\left\{\ket{\varphi_{+|b}(\phi)}, \ket{\varphi_{-|b}(\phi)}\right\}$ (see \eqref{eq:tiltm}). Now, the vectors $\va$ and $\vb$ are given by \eqref{eq:rotated_ab} and obtained from the unit vectors $\vu_a$ and $\vu_b$ of \eqref{eq:en_qubit2} after a rotation by an angle $\phi$ around $\ve_1$ (see Fig.~\ref{fig:axis}). By \eqref{eq:main_2}, the quantity $\chi(\en;\A_\va,\B_\vb)$ constitutes an upper bound for the incompatibility generalized robustness $\eta^{\rm g}(\A_\va,\B_\vb)$. The experimental data are indicated by the blue dots and reported with their error bars, while the solid lines between one point and another are only guides for the reader's eye. The dashed line represents $\eta^{\rm g}(\A_\va,\B_\vb)$ for the chosen measurements. It does not depend on $\phi$ and coincides with $\chi(\en;\A_\va,\B_\vb)$ only for $\phi=0$.}\label{fig:sx}
\end{figure}

In the second part of the experiment, we consider the projective measurements $\A_\va$ and $\B_\vb$, in which $\va$ and $\vb$ are the unit vectors \eqref{eq:rotated_ab}. These vectors are obtained from the unit vectors $\vu_a$ and $\vu_b$ of \eqref{eq:en_qubit2} after a rotation by an angle $\phi$ around $\ve_1$ (see Fig.~\ref{fig:axis}). Since the generalized incompatibility robustness is invariant under unitary conjugation, the value of $\eta^{\rm g}(\A_\va,\B_\vb)$ coincides with $\eta^{\rm g}(\A_{\vu_a},\B_{\vu_b})$ and does not depend on the angle $\phi$. The quantity $\chi(\en;\A_\va,\B_\vb)$, however, varies with $\phi$, as it is confirmed by its experimental values depicted in Fig.~\ref{fig:sx} for a range of angles around $\phi=0$. 
Therefore, according to \eqref{eq:main_2}, in the present situation  $\chi(\en;\A_\va,\B_\vb)$ only provides an upper bound for the incompatibility generalized robustness $\eta^{\rm g}(\A_\va,\B_\vb)$. Such a bound is attained solely in the special case with $\phi=0$, which thus represents an optimal point for the assessment of $\eta^{\rm g}(\A_\va,\B_\vb)$ within the considered family of measurements.

\subsection{The case of unequally noisy measurements}

Up to now, we only considered pairs of equally noisy qubit measurements. For these measurements, we managed to find the analytic expression of the incompatibility generalized robustness $\eta^{\rm g}(\A_{\alpha\va},\B_{\beta\vb})$, so that we could compare it with the upper bound $\chi(\en;\A_{\alpha\va},\B_{\beta\vb})$ and then determine the latter bound in our experiments. Now, we move one step further, and deal with a situation where we lack an analytic formula for $\eta^{\rm g}(\A_{\alpha\va},\B_{\beta\vb})$, but we can still use the experimental estimate of $\chi(\en;\A_{\alpha\va},\B_{\beta\vb})$ to provide an upper bound via inequality \eqref{eq:main_2}.

As in the first part of the experiment, we consider two noisy measurements $\A_{\alpha\va}$ and $\B_{\beta\vb}$ with $\va = \cos\theta\,\ve_1 + \sin\theta\,\ve_2$ and $\vb = \cos\theta\,\ve_1 - \sin\theta\,\ve_2$, but now we no longer assume that the noise parameters $\alpha$ and $\beta$ are equal.
Also in this case, the upper bound $\chi(\en;\A_{\alpha\va},\B_{\beta\vb})$ defined in \eqref{eq:chiab} can be easily calculated, and one finds that the result is still given by \eqref{eq:tpr} up to substituting $\gamma=(\alpha+\beta)/2$. In particular, $\chi(\en;\A_{\alpha\va},\B_{\beta\vb})$ only depends on the angle $\theta$ and the sum of the noise parameters $\alpha+\beta$. In Fig.~\ref{fig:unequal}, we report the theoretical prediction of $\chi(\en;\A_{\alpha\va},\B_{\beta\vb})$ as a function of $\theta$ and $\gamma$ and we compare it with the experimental data. The lozenges represent the result of the experiment for several values of $\alpha$ and $\beta$. The fact that lozenges with different colors overlap confirms that $\chi(\en;\A_{\alpha\va},\B_{\beta\vb})$ depends on the noise parameters $\alpha$ and $\beta$ only through the sum $\alpha+\beta$.
The essential point, however, is that $\chi(\en;\A_{\alpha\va},\B_{\beta\vb})$ constrains the unknown value of $\eta^{\rm g}(\A_{\alpha\va},\B_{\beta\vb})$, so that the experimental determination of the former quantity provides an empirical upper bound for the latter one. In turn, since by definition $\eta^{\rm g}(\A_{\alpha\va},\B_{\beta\vb})$ is itself the upper bound \eqref{eq:def_eta}, determining $\chi(\en;\A_{\alpha\va},\B_{\beta\vb})$ provides a limit to the visibility of the pair $(\A_{\alpha\va},\B_{\beta\vb})$ in the set of all pairs of measurements that are compatible.

\begin{figure}[t]
 \includegraphics[width=0.6\columnwidth]{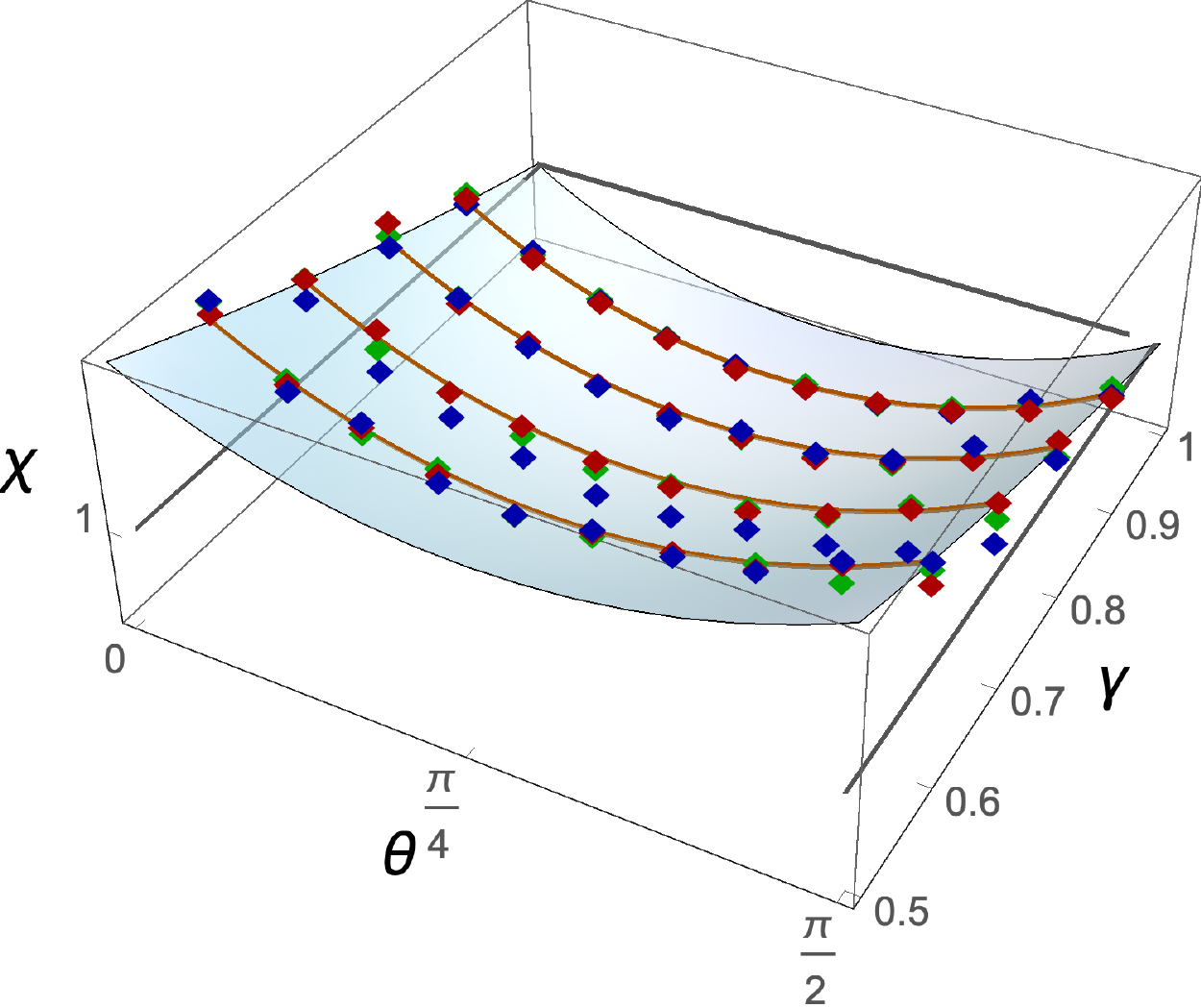}
\caption{The quantity $\chi(\en;\A_{\alpha\va},\B_{\beta\vb})$ for $\va$ and $\vb$ as in \eqref{eq:AB_qubit2}, but with possibly different noise parameters $\alpha$ and $\beta$. In the graph, $\chi(\en;\A_{\alpha\va},\B_{\beta\vb})$ is plotted as a function of $\theta$ and $\gamma=(\alpha+\beta)/2$.
The experimental values are represented by the lozenges, 
while the semitransparent surface is the theoretical prediction \eqref{eq:tpr}. The solid lines are sections for $\gamma$ equal to $0.5$, $0.7$ and $0.9$.
Lozenges corresponding to different values of $\alpha$ and $\beta$ which sum to the same $\gamma$ almost overlap, as we expect since $\chi(\en;\A_{\alpha\va},\B_{\beta\vb})$ depends on $\alpha$ and $\beta$ only through the sum $\alpha+\beta$. For $\alpha=\beta=\gamma$, the theoretical prediction is the same as in Fig.~\ref{fig:equal}.
The quantity $\chi(\en;\A_{\alpha\va},\B_{\beta\vb})$ provides an upper bound for the unknown value of $\eta^{\rm g}(\A_{\alpha\va},\B_{\beta\vb})$.}
\label{fig:unequal}
\end{figure}

\section{Conclusion}\label{sec:conc}

We have investigated both theoretically and experimentally the assessment of incompatibility for two qubit measurements, and to this aim we have employed a strategy based on the success statistics of a state-discrimination task with intermediate information.
Our approach is applicable to a broad and diversified class of measurements related to the photon polarization degrees of freedom in an optical setup. The experimental implementation required us to have full control on the phase and purity in the projection stage of the measurement.
In particular, we have demonstrated that
the generalized incompatibility robustness can be empirically quantified by the proposed protocol
in the presence of couples of unbiased measurements along perpendicular directions in an optimal plane.
This is the case for both projective and noisy measurements, as long as they are characterized by an equal amount of noise.
Furthermore, we have shown that, in the more general scenario with unequal noises, the empirical success probability that can be experimentally accessed in our setup
still provides a significant upper bound for the incompatibility generalized robustness. The data therefore allow to identify a region where a given pair of measurements are incompatible, and to constrain the distance of this pair from the set of all compatible measurements.

Our results prove the effectiveness of a general strategy to investigate both qualitatively and quantitatively
the incompatibility of quantum measurements.
Future studies will naturally address a larger number of measurements and more complex physical systems, aiming at a comprehensive
understanding and control of this intrinsically quantum feature.

\acknowledgments
A.S. and B.V. acknowledge
support from UniMi, via Transition Grant
H2020 and PSR-2 2020.

\appendix
\renewcommand{\theequation}{A\arabic{equation}}
\setcounter{equation}{0}

\section*{Appendix}

In this appendix, we evaluate the incompatibility generalised robustness for the pair of qubit measurements $(\A_{\alpha\va},\B_{\beta\vb})$ defined in \eqref{eq:AB_qubit1} under the assumption of equal noise parameters. Throughout the section, we let $2\theta\in (0,\pi)$ be the angle between the two unit vectors $\va$ and $\vb$, and $\gamma\in (0,1]$ be the common value of the noise parameters $\alpha$ and $\beta$. We first employ the full symmetry properties of the pair $(\A_{\gamma\va},\B_{\gamma\vb})$ in order to drastically simplify the measurements $(\Cc,\D)$ involved in the maximization problem \eqref{eq:def_eta}. Then, we use the well-known characterization of compatibility for dichotomic qubit measurements in order to conclude the calculation of $\eta^{\rm g}(\A_{\gamma\va},\B_{\gamma\vb})$. In this way, we provide the proof of \eqref{eq:main_1}, \eqref{eq:tpr} for $\gamma$ assuming the values \eqref{eq:gamma}.

The result of this appendix expands that of \cite[Section 3.5]{DeFaKa19}, where the value of $\eta^{\rm g}(\A_{\gamma\va},\B_{\gamma\vb})$ was derived by means of semidefinite programming in the noiseless case with $\gamma = 1$.

\subsection{Symmetric measurement assemblages}

The symmetry properties of the two qubit measurements \eqref{eq:AB_qubit1} involve permutations of their outcomes as well as of the measurements themselves. Such a kind of symmetries are easily managed within the general approach described in \cite{NgDeBaGu20}, which applies to tuples of measurements that are transformed by some group action. In order to briefly illustrate this approach, we recall that a {\em finite fiber bundle} is a triple $(Z,T,\pi)$, in which $Z$ and $T$ are finite sets and $\pi:Z\to T$ is a surjective map. The set $Z$ is the {\em total space} of the bundle, $T$ is its {\em base space} and $\pi$ is the {\em bundle projection}. Further, for all $t\in T$, the inverse image $Z_t = \pi^{-1}(\{t\})$ is the {\em fiber} at the base point $t$. A {\em measurement assemblage} on $(Z,T,\pi)$ is then a map $\M:Z\to M_d(\Cb)$, where $M_d(\Cb)$ is the linear space of all $d\times d$ complex matrices, such that
\begin{enumerate}[(i)]
\item $\M(z)\geq 0$ for all $z\in Z$;\label{it:ass_1}
\item $\sum_{z\in Z_t} \M(z) = \id$ for all $t\in T$.\label{it:ass_2}
\end{enumerate}
Properties \eqref{it:ass_1} - \eqref{it:ass_2} mean that, for all base points $t\in T$, the restriction $\M_t = \left.\M\right|_{Z_t}$ is a measurement with outcomes in the fiber $Z_t$.

A measurement assemblage on $(Z,T,\pi)$ is {\em compatible} if all the measurements $(\M_t)_{t\in T}$ are compatible. This is more conveniently stated by saying that there is a measurement $\J$ with outcomes in the set $\Sec(Z,T,\pi)$ of all the sections of the bundle,
$$
\Sec(Z,T,\pi) = \left\{s:T\to Z\mid \pi(s(t)) = t \text{ for all } t\in T\right\}\,,
$$
such that
$$
\M(z) = \sum_{s\in \Sec(Z,T,\pi)} \delta_{z,s(\pi(z))}\,\J(s)
$$
for all $z\in Z$.

Symmetries of a measurement assemblage $\M$ are described by a finite group $G$, acting both on the bundle $(Z,T,\pi)$ and on the Hilbert space $\hh$ of the system, in such a way that $\M$ intertwines the two actions. More precisely, we assume that $(Z,T,\pi)$ is a {\em $G$-bundle}, i.e., there are two left actions $G\times Z\ni (g,z)\mapsto gz\in Z$ and $G\times T\ni (g,t)\mapsto gt\in T$ of the group $G$ on the sets $Z$ and $T$, respectively, such that $\pi(gz) = g\pi(z)$ for all $g\in G$ and $z\in Z$. Moreover, we fix a projective unitary representation of $G$ on $\Cb^d$, that is, a map $U:G\to M_d(\Cb)$ which satisfies $U(g)^*U(g) = \id$ and $U(gh)\propto U(g)U(h)$ for all $g,h\in G$. The two actions of $G$ on $(Z,T,\pi)$ and on $\hh$ combine together, yielding the following action on any measurement assemblage $\M:Z\to M_d(\Cb)$:
$$
g\M(z) = U(g)\M(g^{-1}z)U(g)^*
$$
for all $g\in G$ and $z\in Z$. The measurement assemblage $\M$ is {\em $G$-covariant} if
\begin{enumerate}[(i)]\setcounter{enumi}{2}
\item $g\M = \M$ for all $g\in G$.\label{it:ass_3}
\end{enumerate}
Property \eqref{it:ass_3} implies that $\M_{gt}(gz) = U(g)\M_t(z)U(g)^*$ for all $g$, $t$ and $z$, and, in particular, that the measurement $\M_t$ is covariant in the usual sense with respect to the subgroup $G_t = \{g\in G\mid gt = t\}$ which stabilizes the fiber at $t$ (see \cite[Chapter IV]{PSAQT82}).

For a measurement assemblage $\M$ on an arbitrary finite fiber bundle $(Z,T,\pi)$, we can define the incompatibility generalised robustness $\eta^{\rm g}(\M)$ by obviously extending \eqref{eq:def_eta},
\begin{equation}\label{eq:def_eta_2}
\begin{aligned}
\eta^{\rm g}(\M) = \max\,\{\eta\in [0,1] \mid & \ \eta\,\M + (1-\eta)\,\N\in\JMeas(Z,T,\pi) \\
& \hspace{2cm} \text{ for some } \N\in\Meas(Z,T,\pi)\} \,,
\end{aligned}
\end{equation}
where $\Meas(Z,T,\pi)$ is the set of all measurement assemblages on $(Z,T,\pi)$ and $\JMeas(Z,T,\pi)$ denotes the subset of all assemblages that are compatible.

The following standard symmetrization trick considerably simplifies the evaluation of $\eta^{\rm g}(\M)$ in the covariant case.

\begin{proposition}\label{prop:covariant}
If $\M$ is a $G$-covariant measurement assemblage on $(Z,T,\pi)$, then
\begin{equation*}
\begin{aligned}
\eta^{\rm g}(\M) = \max\,\{\eta\in [0,1] \mid & \ \eta\,\M + (1-\eta)\,\N\in\JMeas(Z,T,\pi) \\
& \hspace{2cm} \text{ for some $G$-covariant } \N\in\Meas(Z,T,\pi)\} \,.
\end{aligned}
\end{equation*}
\end{proposition}
\begin{proof}
We first observe that the set
$$
\{\eta\in [0,1] \mid \eta\,\M + (1-\eta)\,\N\in\JMeas(Z,T,\pi) \text{ for some } \N\in\Meas(Z,T,\pi)\}
$$
coincides with the closed interval $[0,\eta^{\rm g}(\M)]$ by an easy convexity and compactness argument (see e.g.~\cite[pp.~4-5]{DeFaKa19}). Thus, $\eta\leq\eta^{\rm g}(\M)$ if and only there exists $\N\in\Meas(Z,T,\pi)$ such that
$$
\mathsf{O} = \eta\,\M + (1-\eta)\,\N \in \JMeas(Z,T,\pi) \,.
$$
By defining
$$
\mathsf{O}^G = \frac{1}{|G|} \sum_{g\in G} g\mathsf{O}\,,\qquad\qquad \N^G = \frac{1}{|G|} \sum_{g\in G} g\N\,,
$$
where $|G|$ denotes the cardinality of $G$, the convexity of $\Meas(Z,T,\pi)$ and $\JMeas(Z,T,\pi)$ entails that $\mathsf{O}^G$ and $\N^G$ are $G$-covariant measurement assemblages, and
$$
\mathsf{O}^G = \eta\,\M + (1-\eta)\,\N^G \in \JMeas(Z,T,\pi) \,.
$$
The claim then follows.
\end{proof}

\subsection{Incompatibility generalized robustness of two equally noisy unbiased qubit measurements}

Now we show that the qubit measurements \eqref{eq:AB_qubit1} are a particular instance of a $G$-covariant measurement assemblage. We choose $Z=\{+\va,-\va,+\vb,-\vb\}$ as the total space. Further, we let $T=\{a,b\}$ be the base space and define the bundle projection $\pi(\pm\va) = a$ and $\pi(\pm\vb) = b$. The measurements \eqref{eq:AB_qubit1} are then collected in the measurement assemblage
\begin{equation}\label{eq:assemblage_AB}
\M(\pm\va) = \A_{\gamma\va}(\pm)\,,\qquad\qquad \M(\pm\vb) = \B_{\gamma
\vb}(\pm) \,,
\end{equation}
which satisfies $\M_a = \A_{\gamma\va}$ and $\M_b = \B_{\gamma\vb}$. The symmetry group of $\M$ is the dihedral group $G\subset SO(3)$, which consists of the identity element $I$ of $\Rb^3$ together with the three $180^\circ$ rotations $R_1$, $R_2$ and $R_3$ fixing the orthogonal unit vectors $\ve_1 = (\va+\vb)/\|\va+\vb\|$, $\ve_2 = (\va-\vb)/\|\va-\vb\|$ and $\ve_3 = \ve_1\wedge\ve_2$, respectively. Indeed, $G$ acts on the bundle $(Z,T,\pi)$ in the natural way, and it acts on the Hilbert space $\Cb^d$ by means of the projective unitary representation $U(I) = \id$ and $U(R_i) = \sigma_i$ for $i=1,2,3$. It is then easy to see that $\M$ is $G$-covariant.

By Proposition \ref{prop:covariant}, we can evaluate $\eta^{\rm g}(\M)$ by assuming that in the maximum \eqref{eq:def_eta_2} the measurement assemblages $\N$ are $G$-covariant for the dihedral group $G$. Any such assemblage is necessarily of the form
\begin{equation}\label{eq:N1}
\N(\pm\va) = \tfrac{1}{2}\big(\id\pm\vc\cdot\vsigma\big) = \Cc(\pm)\,, \qquad\qquad \N(\pm\vb) = \tfrac{1}{2}\big(\id\pm\vd\cdot\vsigma\big) = \D(\pm)
\end{equation}
with
\begin{equation}\label{eq:N2}
\vc = \nu\,\big(\cos\varphi\,\ve_1 + \sin\varphi\,\ve_2\big)\,,\qquad\qquad\vd = \nu\,\big(\cos\varphi\,\ve_1 - \sin\varphi\,\ve_2\big)
\end{equation}
and $\varphi\in (0,2\pi]$, $\nu\in [0,1]$. Indeed, $\N(+\va) = (\gamma\id+\vc\cdot\vsigma)/2$ for some $\gamma\in\Rb$ and $\vc\in\Rb^3$, where the positivity of $\N(+\va)$ is equivalent to $\no{\vc}\leq\gamma$. Then,
\begin{align*}
\N(-\va) = \N(+R_3\,\va) = U(R_3)\,\N(+\va)\,U(R_3)^* = \tfrac{1}{2} \big(\gamma\id + R_3\,\vc \cdot\vsigma\big)\,.
\end{align*}
The condition $\N(+\va) + \N(-\va) = \id$ implies $\gamma=1$ and $R_3\,\vc = -\vc$, which yield the first two halves of \eqref{eq:N1}-\eqref{eq:N2}. The other two halves follow from
$$
\N(\pm\vb) = \N(\pm R_1\,\va) = U(R_1)\,\N(\pm \va)\,U(R_1)^* = \tfrac{1}{2}\big(\id\pm R_1\,\vc\cdot\vsigma\big) \,.
$$

In order to evaluate $\eta^{\rm g}(\M)$, now we recall from \cite{Busch86} that two unbiased qubit measurements $\E(\pm) = (\id\pm\ve\cdot\vsigma)/2$ and $\F(\pm) = (\id\pm\vf\cdot\vsigma)/2$ are compatible if and only if
$$
\big\|\ve+\vf\big\| + \big\|\ve-\vf\big\| \leq 2\,.
$$
From \eqref{eq:AB_qubit1}, \eqref{eq:AB_qubit2}, \eqref{eq:assemblage_AB}, \eqref{eq:N1} and \eqref{eq:N2}, it then follows that the measurement assemblage $\eta\,\M + (1-\eta)\,\N$ is compatible if and only if
\begin{equation*}
\mo{\eta\gamma\cos\theta + (1-\eta)\nu\cos\varphi} + \mo{\eta\gamma\sin\theta + (1-\eta)\nu\sin\varphi} \leq 1\,,
\end{equation*}
or equivalently
$$
\big\|\eta\,\va + (1-\eta)\,\vc\big\|_1 \leq 1 \,,
$$
where $\no{\cdot}_1$ is the $\ell^1$-norm of $\Rb^3$, i.e., $\no{\vu}_1 = \mo{u_1}+\mo{u_2}+\mo{u_3}$.

All the above discussion is summarized by the relation
\begin{equation}\label{eq:eta_simple}
\eta^{\rm g}(\A_{\gamma\va},\B_{\gamma\vb}) = \max\,\{\eta\in [0,1] \mid \ \eta\,\va + (1-\eta)\,\vc\in Q \text{ for some } \vc\in D\} \,,
\end{equation}
where $D=\{\vu\in\Rb^3 \mid \no{\vu}\leq 1,\ u_3=0\}$ and $Q=\{\vu\in\Rb^3 \mid \no{\vu}_1\leq 1,\ u_3=0\}$ are the unit disk and the unit square in the plane spanned by $\ve_1$ and $\ve_2$, respectively. The condition $\no{\va}_1 = \gamma\,(\cos\theta+\sin\theta)\leq 1$ is equivalent to the compatibility of $\A_{\gamma\va}$ and $\B_{\gamma\vb}$, that is, $\eta^{\rm g}(\A_{\gamma\va},\B_{\gamma\vb}) = 1$. Thus, from now on we only consider the case with $\gamma\,(\cos\theta+\sin\theta)>1$.

\begin{figure}[!ht]
\begin{tikzpicture}[domain=0:4, scale=2.5]

\def\ax{0.65};\def\ay{0.60};
\def\cpx{-0.45};\def\cpy{-0.893029}; 
\def\cx{-0.2};

\draw[thick,fill= darkgreen!50!white] (0,0) circle (1);
\draw[thick,fill=blue!50!white] (1,0)--(0,1)--(-1,0)--(0,-1)--cycle;

\draw (-0.3,0.3)node{$Q$};
\draw (-0.6,0.6)node{$D$};

\draw[->, thick](-1.25,0)--(1,0)node[anchor=north west]{$\ve_1$}--(1.25,0);
\draw[fill=black] (1,0) circle (0.015);

\draw[->, thick](0,-1.25)--(0,1)node[anchor=south east]{$\ve_2$}--(0,1.25);
\draw[fill=black] (0,1) circle (0.015);

\draw (0,0)node[anchor=north east]{$0$};
\draw[fill=black] (0,0) circle (0.015);

\draw (\ax,\ay)node[anchor=south east]{$\va$};
\draw[fill=black] (\ax,\ay) circle (0.015);

\draw (\cpx,\cpy)node[anchor=north east]{$\vc^{\,\prime}$};
\draw[fill=black] (\cpx,\cpy) circle (0.015);

\draw (\cpx,\cpy)--(\ax,\ay);

\draw (\cx,{((\ay-\cpy)/(\ax-\cpx))*(\cx-\cpx)+\cpy})node[anchor=south east]{$\vc$};
\draw[fill=black] (\cx,{((\ay-\cpy)/(\ax-\cpx))*(\cx-\cpx)+\cpy}) circle (0.015);

\draw ({(\cpx*(\ay-1)+\ax*(1-\cpy))/(\ax+\ay-\cpx-\cpy)},{1-(\cpx*(\ay-1)+\ax*(1-\cpy))/(\ax+\ay-\cpx-\cpy)})node[anchor=east]{$\va^{\,\prime}$};
\draw[fill=black] ({(\cpx*(\ay-1)+\ax*(1-\cpy))/(\ax+\ay-\cpx-\cpy)},{1-(\cpx*(\ay-1)+\ax*(1-\cpy))/(\ax+\ay-\cpx-\cpy)}) circle (0.015);

\end{tikzpicture}

\caption{The set of all equally noisy pairs of qubit measurements (green disk $D$) and the subset of all compatible pairs (blue square $Q$), with the geometric construction of the robustness \eqref{eq:eta_simple}.\label{fig:DQ}}
\end{figure}
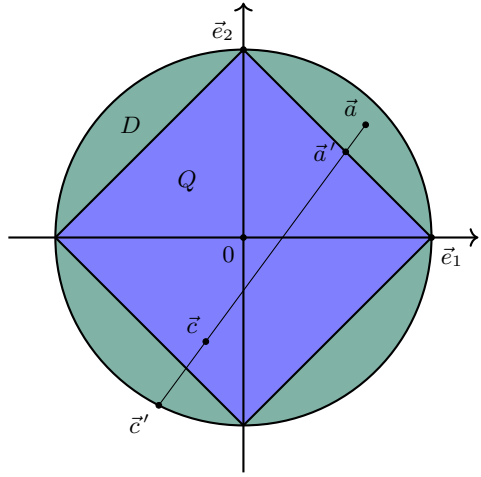

In the maximum \eqref{eq:eta_simple}, we can assume with no restriction that the optimal $\vc$ is given by \eqref{eq:N2} with $\nu=1$ and $\varphi\in [\pi/2,\,2\pi]$. Indeed, if this is not the case, then either $\nu=1$ and $\varphi\in (0,\,\pi/2)$, which is impossible (see Fig.~\ref{fig:DQ}), or $\nu<1$, which implies that $\eta$ is not maximal because we have $\vc^{\,\prime}=(1+\varepsilon)\,\vc - \varepsilon\,\va\in D$ for sufficiently small $\varepsilon > 0$ and then $\eta'\,\va + (1-\eta')\,\vc^{\,\prime}\in Q$ with $\eta'=(\eta+\varepsilon)/(1+\varepsilon) > \eta$. For fixed $\vc$, the maximal $\eta$ is attained when the vector $\va^{\,\prime} = \eta\,\va + (1-\eta)\,\vc$ lies in the upper right edge of the square $Q$, that is, $a^{\prime}_1 + a^{\prime}_2 = 1$, which gives
$$
\eta = \frac{1-(\cos\varphi+\sin\varphi)}{\mu(\cos\theta+\sin\theta)-(\cos\varphi+\sin\varphi)} = 1+\frac{\frac{1}{\sqrt{2}}-\mu\cos\left(\theta-\frac{\pi}{4}\right)}{\mu\cos\left(\theta-\frac{\pi}{4}\right) - \cos\left(\varphi-\frac{\pi}{4}\right)} \,.
$$
Maximizing over $\varphi\in [\pi/2,\,2\pi]$, we obtain
\begin{equation*}
\eta^{\rm g}(\A_\mu,\B_\mu) = \frac{\sqrt{2}+1}{\sqrt{2}+\mu\,(\cos\theta+\sin\theta)} \,.
\end{equation*}


\end{document}